\documentclass[a4paper,12pt]{article}
\pdfoutput=1 

\usepackage[utf8]{inputenc}
\usepackage[english]{babel}		

\usepackage[backend=biber,sorting=nyt,style=numeric-comp,doi=false,firstinits=true,date=year,isbn=false,maxbibnames=9]{biblatex} 
\usepackage{csquotes}

\usepackage[left=2.5cm,right=2.5cm,top=2.5cm,bottom=2.5cm]{geometry} 
\renewbibmacro{in:}{ 
    \ifentrytype{article}{}{%
    \printtext{\bibstring{in}\intitlepunct}}}

\usepackage{wrapfig}
\usepackage{setspace}
\usepackage{hhline}

\usepackage{amssymb}
\usepackage{amsmath}
\usepackage{amsfonts}   
\usepackage{mathtools}
\usepackage{mathrsfs}				
\usepackage{amsbsy}

\usepackage{amsthm}				
\newtheoremstyle{owntheorem}{}
  {}
  {\itshape}
  {}
  {}
  {:}
  {1em}
  {\textbf{\thmname{#1}\thmnumber{ #2}}\thmnote{ [#3]}}
\newtheoremstyle{owndefinition}%
  {}{}%
  {}{}%
  {}{:}%
  {1em}%
  {\textbf{\thmname{#1}\thmnumber{ #2}}\thmnote{ [#3]}}
\theoremstyle{owntheorem}
\newtheorem{theorem}{Theorem}[section]
\newtheorem{lemma}[theorem]{Lemma}
\newtheorem{cor}[theorem]{Corollary}

\theoremstyle{owndefinition}
\newtheorem{defn}{Definition}[section]
\newtheorem{example}{Example}[section]
\newtheorem*{remark}{Remark}
\newtheorem*{acknowledgements}{Acknowledgements}

\usepackage{url} 				
\usepackage{wasysym}				
\usepackage{enumerate}
\usepackage{rotating}				
\usepackage{multirow}
\usepackage{array}
\usepackage{longtable}
\usepackage{lscape}

\usepackage{tikz}
\usepackage{yfonts}
\usepackage{tablefootnote}
\usepackage{bbm}
\usepackage{wasysym}
\usepackage{blkarray}
\usepackage{authblk}

\numberwithin{equation}{section}

\renewcommand\Re{\operatorname{Re}} 		

\renewcommand{\d}[1][]{\mathrm{d}#1\,} 

\newcommand{\comma}{\quad \textrm{ ,}} 
\newcommand{\point}{\quad \textrm{ .}} 
\newcommand{\monthword}[1]{\ifcase#1\or \or \or \or April\fi}

\newcommand{\HypF}[3]{ {}_2F_1 \left.\left(\genfrac{}{}{0pt}{}{#1}{#2}\right| #3 \right) }

\newcommand{\vol}{\operatorname{vol}}

\newcommand{\rank}{\operatorname{rank}}
\newcommand{\Aa}{\mathcal{A}}
\newcommand{\Aas}{\mathcal{A}_\sigma}
\newcommand{\Aabs}{\mathcal{A}_{\bar{\sigma}}}
\newcommand{\nuu}{{\underline{\nu}}}

\newcommand{\Conv}{\operatorname{Conv}}
\newcommand{\relint}{\operatorname{relint}}

\newcommand{\order}[1]{\mathcal{O} \left( #1 \right)}

\newcommand{\spacepar}{\vspace{\baselineskip}}

\title{Hypergeometric Series Representations of Feynman Integrals by GKZ Hypergeometric Systems}
\author{René Pascal Klausen\footnote{Department of Physics at Humboldt University of Berlin and Max-Planck-Institute for Gravitational Physics (Albert-Einstein-Institute) in Potsdam-Golm (klausen@physik.hu-berlin.de)}}

\date{May 27, 2020}

\begin{document}
  \maketitle
  
  \begin{abstract}
    We show that almost all Feynman integrals as well as their coefficients in a Laurent series in dimensional regularization can be written in terms of Horn hypergeometric functions. By applying the results of Gelfand-Kapranov-Zelevinsky (GKZ) we derive a formula for a class of hypergeometric series representations of Feynman integrals, which can be obtained by triangulations of the Newton polytope $\Delta_G$ corresponding to the Lee-Pomeransky polynomial $G$. Those series can be of higher dimension, but converge fast for convenient kinematics, which also allows numerical applications. Further, we discuss possible difficulties which can arise in a practical usage of this approach and give strategies to solve them.
  \end{abstract}



  \section*{Introduction}


  In the early days of calculating Feynman amplitudes, it was proposed by Regge to consider Feynman integrals as a kind of generalized hypergeometric functions \cite{ReggeAlgebraicTopologyMethods1967}, where the singularities of those hypergeometric functions coincide with the Landau singularities. Later on Kashiwara and Kawai \cite{Kashiwara.KawaiHolonomicSystemsLinear1976} showed that Feynman integrals satisfy indeed holonomic differential equations, where the singularities of those holonomic differential equations are determined by the Landau singularities.
    
  Apart from characterizing the Feynman integral by ``hypergeometric'' partial differential equation systems, many applications determine the Feynman integral as a generalized hypergeometric series. Usually, the often used Mellin-Barnes approach \cite{SmirnovFeynmanIntegralCalculus2006} results in Pochhammer series ${}_pF_q$, Appell functions, Lauricella functions and related functions by applying the residue theorem \cite{Boos.DavydychevMethodEvaluatingMassive1991}. Furthermore, for arbitrary one-loop Feynman integrals it is known that they can always be represented by a small set of hypergeometric series \cite{Fleischer.Jegerlehne.TarasovNewHypergeometricRepresentation2003}.
  Thirdly, the Feynman integral may be expressed by ``hypergeometric'' integrals like the generalized Meijer $G$- or Fox $H$-functions \cite{Buschman.SrivastavaFunctionAssociatedCertain1990, Inayat-HusNewPropertiesHypergeometric1987, Inayat-HusNewPropertiesHypergeometric1987a}.
    
  Thus there arise three different notions of the term ``hypergeometric'' in the Feynman integral calculus, where every notion generalizes different characterizations of the classical hypergeometric Gauß function ${}_2F_1(a,b,c;x)$. In the late 1980s Gelfand, Kapranov, Zelevinsky (GKZ) and collaborators \cite{Gelfand.Kapranov.ZelevinskyGeneralizedEulerIntegrals1990, Gelfand.Graev.RetakhGeneralHypergeometricSystems1992, Gelfand.Graev.ZelevinskiHolonomicSystemsEquations1987, Gelfand.Kapranov.ZelevinskiHypergeometricFunctionsToric1991, GelfandPartGeneralTheory1989, Gelfand.Kapranov.ZelevinskyDiscriminantsResultantsMultidimensional2008a} were starting to develop a comprehensive method to generalize the notion of ``hypergeometric'' functions in a consistent way. Those functions are called $\Aa$-hypergeometric functions and are defined by a special holonomic system of partial differential equations.
    
  As Gelfand, Kapranov and Zelevinsky illustrated with Euler integrals, the GKZ approach not only generalizes the concept of hypergeometric functions but can also be used for analyzing and solving integrals \cite{Gelfand.Kapranov.ZelevinskyGeneralizedEulerIntegrals1990}.
    
  For physicists the GKZ perspective is not entirely new. Already in the 1990s, string theorists applied the GKZ approach in order to calculate period integrals and worked out the mirror symmetry \cite{Hosono.Klemm.Theisen.YauMirrorSymmetryMirror1995, Hosono.Klemm.Theisen.YauMirrorSymmetryMirror1995a}. Recently, the GKZ approach was also used to obtain differential equations for the Feynman integral from the maximal cut \cite{VanhoveFeynmanIntegralsToric2018}. Still, the approach of Gelfand, Kapranov and Zelevinsky is no common practice among physicists.
    
  In 2016 Nasrollahpoursamami showed that the Feynman integral satisfies a differential equation system which is isomorphic to a GKZ system \cite{NasrollahpPeriodsFeynmanDiagrams2016}. Very recently\footnote{The present paper was developed independently from \cite{delaCruzFeynmanIntegralsAhypergeometric2019}, which the author noticed just shortly before publication.}, this fact was also shown directly in \cite{delaCruzFeynmanIntegralsAhypergeometric2019} based on the Lee-Pomeransky representation of the Feynman integral.
    
  Beyond the above statement which characterizes generalized Feynman integrals as $\Aa$-hypergeometric functions, we show that generalized Feynman integrals as well as every coefficient in the $\epsilon$ expansion in dimensional regularization belong to the class of Horn hypergeometric functions. Furthermore, we give an explicit formula for a multivariate series representation of a generalized Feynman integral for unimodular triangulations. This allows, to evaluate the Feynman integral efficiently for convenient kinematic regions.
    
  Therefore, from the perspective of GKZ it turns out that Horn hypergeometric functions and Feynman integrals share many properties, e.g. Horn hypergeometric functions satisfy special relations similar to the IBP-relations of the Feynman integral \cite{Kniehl.TarasovFindingNewRelationships2012}. In this article, we work out the connection between Feynman integrals and Horn hypergeometric functions from the GKZ perspective and we give a strategy, as to how to use this knowledge in the evaluation of Feynman integrals.
    
  The connections between Feynman integrals and hypergeometric functions was investigated over decades and a comprehensive summary of these investigations can be found in \cite{Kalmykov.Kniehl.Ward.YostHypergeometricFunctionsTheir2008}. Horn hypergeometric functions also appear often in the Mellin-Barnes approach and have been studied intensively by other authors, e.g \cite{Bytev.Kniehl.MochDerivativesHorntypeHypergeometric2017, SadykovHypergeometricFunctionsSeveral2002}.
     
  The paper is structured as follows. After recalling the parametric representations of the Feynman integral, which are crucial for this approach, we mention some properties of the Feynman integral, which can be obtained from their algebraic geometry description. This first section will also include a short mathematical interlude about convex polytopes, which are necessary in the description of the properties of Feynman integrals, as well as in the later GKZ approach. Secondly, we briefly introduce the GKZ hypergeometric system, which is a system of partial differential equations and recall a series solution of those systems. This leads us to the last section, in which we merge these two aspects in order to derive an analytical series representation of the Feynman integral. To complete this section we also discuss some features as well as possible difficulties which can arise in the evaluation of Feynman integrals with GKZ systems. We conclude this article by calculating the full massive sunset Feynman diagram as a non-trivial example in order to illustrate this procedure and to give a glimpse of its further scope of application.

  \begin{acknowledgements} 
    This research is supported by the International Max Planck Research School for Mathematical and Physical Aspects of Gravitation, Cosmology and Quantum Field Theory. I would like to express my special thanks to Christian Bogner for his great encouragement, his supervision and his assistance in difficult circumstances. Further, I would like to thank Dirk Kreimer for his support, organizational skills and for including me in his group. I am also grateful to Josua Faller and Tatiana Stahlhut for their proofreading and all the members of Dirk Kreimer's group for helpful discussions.
  \end{acknowledgements}

  
  \section{Feynman Integrals} \label{sec:FI}
  
    
  In this section we shortly recall the parametric representations of Feynman integrals, which are based on graph polynomials \cite{Bogner.WeinzierlFeynmanGraphPolynomials2010}. Those parametric representations contain normally two graph polynomials, known as the Symanzik polynomials. Recently, Lee and Pomeransky \cite{Lee.PomeranskyCriticalPointsNumber2013} introduced a parametric representation, which depends only on one graph polynomial, which simplifies the application of the Gelfand-Kapranov-Zelevinsky approach. In this representation the Feynman integral can be formally described as a multivariate Euler-Mellin integral.
  
  In the second part of this section, we rely on some general properties of Euler-Mellin integrals in order to recall some basic properties of the Feynman integral from the perspective of the Lee-Pomeransky representation. To keep this discussion short, we refer to standard literature \cite{SmirnovFeynmanIntegralCalculus2006, WeinzierlArtComputingLoop2006, StermanIntroductionQuantumField1993} for further properties of the Feynman integral.
  
  Since the formulation of this properties of Feynman integrals, as well as the later introduced GKZ approach, include convex polytopes, we give a short mathematical interlude about convex polytopes between these two parts. A proceeding and more detailed description of polytopes can be found in the appendix.

    \subsection{Parametric Representations of Feynman Integrals}
    The Feynman integral in momentum space is an $(Ld)$-dimensional loop integral over propagators. Since propagators are at most quadratic in the loop momenta, one can rewrite the Feynman integral as an $n$-dimensional integral over Schwinger parameters $x_1,\ldots,x_n\in\mathbb R_{\geq 0}$ \cite{PanzerFeynmanIntegralsHyperlogarithms2015}. With this rephrasing one can make the Feynman integral also meaningful for $d\in\mathbb C$, as required in the procedure of dimensional regularization \cite{tHooft.VeltmanRegularizationRenormalizationGauge1972}. \spacepar
   
    In the parametric version of Feynman integrals, two polynomials in the Schwinger parameters arise, which are known as first and second Symanzik polynomial \cite{Bogner.WeinzierlFeynmanGraphPolynomials2010}
    \begin{align}
      U &= \sum_{T\in\mathcal T_1} \prod_{e_i\notin T} x_i \\
      F &= - \sum_{F\in\mathcal T_2} s_F \prod_{e_i\notin F} x_i + U \sum_{e_i\in E} x_i m_i^2 \point
    \end{align}
    Here, $\mathcal T_i$ is the set of spanning forests of the Feynman graph $\Gamma$ consisting of $i$ components, $E$ denotes the set of all edges of $\Gamma$ and $s_F$ is the squared momentum flowing from the one component of the $2$-forest $F$ to the other component. For Euclidean kinematics the coefficients of all monomials in the second Symanzik polynomial will be positive. In the following we restrict the discussion to Euclidean kinematics in order to avoid zeros of the Symanzik polynomials in the positive orthant $x\in\mathbb R^n_{>0}$. Since almost every Feynman integral with Minkowskian kinematics has a non-vanishing overlap with the Euclidean region, one can extend the Euclidean result to the Minkowskian result by analytic continuation. In this procedure additional divergences can appear and thus the analytic continuation can be far from trivial \cite{PanzerFeynmanIntegralsHyperlogarithms2015}. \spacepar
   
    As the starting point of this discussion we define the Feynman integral in the Feynman parametric representation.
   
    \begin{defn}[Feynman integral] 
      For a given Feynman graph $\Gamma$ with $n$ edges, $L$ loops\footnote{In contrast to the nomenclature in graph theory, in the context of Feynman integrals a ``loop'' means a closed path. Thus, $L$ is the first Betti number of $\Gamma$.} and the Symanzik polynomials $U$ and $F$, the Feynman integral is given as an $n$-dimensional integral
      \begin{align}
        I_\Gamma (\nu,d,p,m) := \frac{\Gamma(\omega)}{\Gamma(\nu)} \int_{\mathbb{R}^n_+} \d[ x]  x^{\nu-1} \delta\left(1-\sum_{i=1}^n x_i\right) \frac{U^{\omega-\frac{d}{2}}}{F^\omega} \label{eq:FI1}
      \end{align}
      where $\omega :=\sum_{i=1}^n \nu_i -\frac{L d}{2}$ is the superficial degree of divergence and $\delta(x)$ denotes the Dirac $\delta$ function. The $\nu\in\mathbb{C}^{n}$ with $\Re\nu_i >0$ are the propagator powers. Note that for simplicity we use a multi-index notation throughout the whole paper, which implies the following shorthand notations
      \begin{align}
        \d[ x]  &:= \, \d[x_1] \cdots \, \d[x_n] \nonumber\\
         x^{\nu-1} &:= x_1^{\nu_1-1} \cdots x_n^{\nu_n-1} \\
        \Gamma(\nu) &:= \Gamma(\nu_1)\cdots\Gamma(\nu_n) \nonumber\point
      \end{align}
    \end{defn}
    Clearly, those integrals converge not for every values of $\nu\in\mathbb C^n$ and $d\in\mathbb C$ and also not for every choice of the polynomials $U$ and $F$. It can be shown, that besides of the class of massless tadpole graphs, these integrals define meromorphic functions in $\nu$ and $d$. Therefore, we will consider the Feynman integral always as the meromorphic continuation of the integral (\ref{eq:FI1}) to the whole complex plane. The convergence as well as the meromorphic expansion will be discussed in more detail in section \ref{sec:FIprop}.
   
    \begin{remark}
      One is often interested in $\nu_i\in\mathbb N$. However, in the following it will be convenient to consider a slightly more general notion of Feynman integrals where the propagator powers $\nu_i$ are not restricted to integer values only. Just the restriction $\Re \nu_i > 0$ is necessary to guaranty the parametric rewriting and convergence. These conditions are often assumed in the literature e.g. in \cite{SpeerGeneralizedFeynmanAmplitudes1969, Bitoun.Bogner.Klausen.PanzerFeynmanIntegralRelations2019}. 
    \end{remark}
   
    The parametric representation in equation (\ref{eq:FI1}) is not the only representation in terms of Schwinger parameters. For the following approach another parametric representation invented by Lee and Pomeransky \cite{Lee.PomeranskyCriticalPointsNumber2013} is more convenient. 
 
    \begin{theorem}[Lee-Pomeransky representation \cite{Lee.PomeranskyCriticalPointsNumber2013}] \label{thm:leepom}
      The Feynman integral from equation (\ref{eq:FI1}) can also be written as
      \begin{align}
        I_\Gamma (\nu,d,p,m) = \frac{ \Gamma\left(\frac{d}{2}\right) }{ \Gamma\left(\frac{d}{2}-\omega\right) \Gamma(\nu)} \int_{\mathbb{R}^n_+} \d[ x]  x^{\nu-1}  G^{-\frac{d}{2}} \comma \label{eq:leepom} 
      \end{align}
      where the integral depends only on the sum of Symanzik polynomials $G=U+F$, which we call the Lee-Pomeransky polynomial $G$. The equality of representations is in the sense of meromorphic extension.
    \end{theorem}
    \begin{proof} (A proof can be found also in \cite{Bitoun.Bogner.Klausen.PanzerFeynmanIntegralRelations2019})
      Since $U$ is homogeneous of degree $L$ and $F$ is homogeneous of degree $L+1$, the integral in (\ref{eq:leepom}) as a function of $d$ converges in the strip \break $\Lambda =  \left\{d\in\mathbb{C} \left|  \frac{2\Re \sum_i \nu_i}{L+1}  < \Re d < \frac{2\Re \sum_i \nu_i}{L}\right\}\right.$. As we consider $\Re \nu_i > 0$ it is $\Lambda\neq\emptyset$. For an equality in the sense of meromorphic extension it is sufficient to show that there is a non-vanishing interval where the equality of representations holds.
      
      Inserting $1=\int_0^\infty \d[s] \delta(s-\sum_i x_i)$ in (\ref{eq:leepom}), changing the integration order and substituting $x_i \rightarrow s x_i$ one obtains
      \begin{align}
        I_\Gamma (\nu,d,p,m) = \frac{ \Gamma\left(\frac{d}{2}\right) }{ \Gamma\left(\frac{d}{2}-\omega\right) \Gamma(\nu)} \int_0^\infty \d[s] \int_{\mathbb{R}^n_+} \d[ x] \delta(s-s\sum_i x_i)  x^{ \nu-1} s^{\sum_i\nu_i}  (s^L U + s^{L+1}F)^{-\frac{d}{2}} \textrm{  .}
      \end{align}
      Since the integral $\int_0^\infty \d[s] s^{\sum_i \nu_i-L\frac{d}{2} - 1} (U+sF)^{-\frac{d}{2}}$ can be calculated explicitly as a beta function in the region $\Lambda$ (and for $U,F>0$) one attains the representation (\ref{eq:FI1}).
    \end{proof}
  
    As a consequence of theorem \ref{thm:leepom} , the whole structure of a Feynman integral can be expressed in only one single polynomial $G\in \mathbb C[x_1,\ldots,x_n]$ . As it is possible for every polynomial, one can write the Lee-Pomeransky polynomial $G$ as
    \begin{align}
      G_{ z}( x) = \sum_{\mathrm{ a}_j\in\mathrm{A}} z_{j}  x^{\mathrm{ a}_j} = \sum_{j=1}^N z_{j} x_1^{\mathrm{a}_{1j}} \ldots x_n^{\mathrm{a}_{nj}} \label{eq:Gzx}
    \end{align}
    where $\mathrm A$ is a finite set consisting of $N$ pairwise distinct column vectors $\mathrm{ a}_j\in\mathbb Z_{\geq 0}^n$ and the coefficients of this polynomial are complex numbers $ z\in(\mathbb C\setminus\{0\})^N$, which contain the kinematics and masses. This representation is unique, up to the possibility of different monomial orderings. Without loss of generality we can fix an arbitrary monomial ordering and we denote the set of column vectors $\mathrm A$ in a matrix structure. As we will observe later, the exponential vectors of the Lee-Pomeransky polynomial $G_z$ satisfy an affine structure. Thus, it will be convenient to define 
    \begin{align}
      \Aa := \begin{pmatrix} 1 \\ \mathrm{A}  \end{pmatrix} = \begin{pmatrix} 1 & 1 & \ldots & 1 \\ \mathrm{a}_1 & \mathrm{a}_2 & \ldots & \mathrm{a}_N  \end{pmatrix} \in \mathbb Z_{\geq 0}^{(n+1)\times N} \point \label{eq:Aa}
    \end{align}
    Further, for an index set $\sigma \subset \{1,\ldots,N\}$ the notation $\Aas$ denotes the restriction of $\Aa$ according to columns indexed by $\sigma$. 
    
    In addition, we define coordinates $\nuu := (\nu_0, \nu)\in \mathbb C^{n+1}$, where the first entry of $\nuu$ contains the physical dimension $\nu_0 := \frac{d}{2}$. From a mathematical point of view, the physical dimension $d$ and the propagator powers $\nu_i$ have a similar role in parametric Feynman integrals. This is the reason why instead of using dimensional regularization one can also regularize the integral by considering the propagator powers as complex numbers as done in analytic regularization \cite{SpeerGeneralizedFeynmanAmplitudes1969}.
    
    In fact, this definition of $G_{ z}( x)$ includes a generalization of the original Feynman integral, since also the first Symanzik polynomial gets coefficients $z_i$. Thus, we define:
    \begin{defn}[Generalized Feynman Integrals] \label{def:genericFI}
      Let $G_{ z}( x)$ a Lee-Pomeransky polynomial with generic coefficients $ z\in \mathbb C^N$ satisfying $\Re z_j >0$. The generalized Feynman integral is the meromorphic continuation of the integral
      \begin{align}
        J_\Aa (\nuu, z) := \Gamma(\nu_0) \int_{\mathbb{R}^n_+} \d[ x]  x^{ \nu-1}  G_{ z}( x)^{-\nu_0} \point  \label{eq:FI-J}
      \end{align}
      defined on $\nuu=(\nu_0, \nu)\in\mathbb C^{n+1}$.
    \end{defn}
    \begin{remark}
      In this definition the complete graph structure, which is necessary to evaluate the Feynman integral, is given by the matrix $\Aa$. The variables $z\in\mathbb C^N$ contain the physical information about kinematics and masses and the $\nuu\in\mathbb C^{n+1}$ are the regularization parameters. 
      
      In contrast to equation (\ref{eq:leepom}) the coefficients in $G_{ z}$ which come from the first Symanzik polynomial are treated as generic, instead of equal to $1$. We will discuss later how one can remove these auxiliary variables afterwards to obtain the ``ordinary'' Feynman integral. To avoid unnecessary prefactors in the following, we omit also the factor $\Gamma(\nu_0-\omega)\Gamma(\nu)$ in this definition
    \end{remark} 
    With this definitions one can derive another representation of the Feynman integral as a multi-dimensional Mellin-Barnes integral. 
  
    \begin{theorem}[Representation as Fox $H$-function] \label{thm:MB}
      Let $\sigma\subset \{1,\ldots,N\}$ be an index subset with cardinality $n+1$, such that the matrix $\Aa$ restricted to columns of $\sigma$ is invertible, $\det\Aas \neq 0$. Then the Feynman integral can be written as the multi-dimensional Mellin-Barnes integral
      \begin{align}
        J_\Aa (\nuu, z) =  \frac{ z_\sigma^{-\Aas^{-1}\nuu}}{|\det\Aas|} \int_\gamma \frac{\d[ t]}{(2\pi i)^r} \Gamma( t) \Gamma(\Aas^{-1}\nuu-\Aas^{-1}\Aabs  t)  z_{\bar\sigma}^{- t}  z_\sigma^{\Aas^{-1}\Aabs  t} \label{eq:MBFox}
      \end{align}
      wherever this integral converges. The set $\bar\sigma:=\{1,\ldots N\} \setminus \sigma$ denotes the complement of $\sigma$, containing $r:=N-n-1$ elements. Restrictions of vectors and matrices to those index sets are similarly defined as $ z_\sigma := (z_i)_{i\in\sigma}$, $ z_{\bar\sigma} := (z_i)_{i\in\bar\sigma}$, $\Aabs := (a_i)_{i\in\bar\sigma}$. Every component of the integration contour $\gamma\in\mathbb C^r$ goes from $-i\infty$ to $i\infty$ such that the poles of the integrand are separated. 
    \end{theorem}
    \begin{cor} \label{cor:Fox}
      Let $N=n+1$ or in other words let $\Aa$ be quadratic. If there is a region $D\subseteq \mathbb C^{n+1}$ such that the Feynman integral $J_\Aa(\nuu, z)$ converges absolutely for $\nuu\in D$, the matrix $\Aa$ is invertible and the Feynman integral is only a simple combination of $\Gamma$-functions:
      \begin{align}
        J_\Aa (\nuu, z) = \frac{\Gamma(\Aa^{-1}\nuu)}{ |\det \Aa|}  z^{-\Aa^{-1}\nuu} \point
      \end{align}
    \end{cor}  
    \begin{proof}
      Starting from equation (\ref{eq:FI-J}) by the Schwinger trick one gets
      \begin{align}
        J_\Aa(\nuu, z) = \int_{\mathbb{R}^{n+1}_+} \d[x_0] x_0^{\nu_0-1} \d[ x]  x^{\nu-1} e^{-x_0 G} \comma
      \end{align}
      where $\nu_0=\frac{d}{2}$. Writing $\underline{x} = (x_0, x)$, $\nuu = (\nu_0,\nu)$ and using the Cahen-Mellin integral representation of exponential function one obtains
      \begin{align}
         J_\Aa(\nuu, z) = \int_{\mathbb{R}^{n+1}_+} \d[\underline{x}] \underline{x}^{\nuu-1} \int_{\delta+i\mathbb R^{n+1}} \frac{\d[u]}{(2\pi i)^{n+1}} \Gamma(u) z_\sigma^{-u} \underline{x}^{-\Aas u} \int_{\delta+i\mathbb R^r} \frac{\d[t]}{(2\pi i)^r} \Gamma(t) z_{\bar\sigma}^{-t} \underline{x}^{-\Aabs t} 
      \end{align}
      with $u\in\mathbb{C}^{n+1}$, $t\in\mathbb{C}^r$, some arbitrary positive numbers $\delta_i >0$ and where we split the polynomial $G$ into a $\sigma$ and a $\bar\sigma$ part. By a substitution $u\rightarrow \Aas^{-1} u^\prime$ it is
      \begin{align}
         J_\Aa(\nuu,z) = |\det \Aas^{-1}| &\int_{\delta+i\mathbb R^{r}} \frac{\d[t]}{(2\pi i)^r} \Gamma(t) z_{\bar\sigma}^{-t} \nonumber \\ 
         & \quad \int_{\mathbb{R}^{n+1}_+} \d[\underline{x}] \int_{\Aas \delta+i \Aas \mathbb R^{n+1}} \frac{\d[u^\prime]}{(2\pi i)^{n+1}}  \Gamma(\Aas^{-1}u^\prime) z_\sigma^{-\Aas^{-1}u^\prime} \underline{x}^{\nuu-u^\prime -\Aabs t- 1} \point
      \end{align}
      Since the matrix $\Aas$ contains only positive values, the integration region remains the same $\Aas \delta + i \Aas \mathbb R^{n+1} \simeq \delta^\prime+i\mathbb R^{n+1}$ with some other positive numbers $\delta^\prime\in\mathbb R^{n+1}_{>0} $, which additionally have to satisfy $\Aas^{-1}\delta^\prime >0$. By Mellin's inversion theorem \cite{AntipovaInversionMultidimensionalMellin2007} only the $t$-integration remains and one obtains equation (\ref{eq:MBFox}).
      
      Thereby, the integration contour has to be chosen, such that the poles are separated from each other in order to satisfy $\Aas^{-1}\delta^\prime>0$. More specific this means that the contour $\gamma$ has the form $c+i\mathbb R^{n+1}$ where $c\in\mathbb R^{n+1}_{>0}$ satisfies $\Aas^{-1}\nuu - \Aas^{-1}\Aabs c>0$. Clearly, in order for those $c$ to exist, the possible values of parameters $\nuu$ are restricted.
      
      The proof of the corollary is a special case, where one does not have to introduce the integrals over $t$. The existence of the inverse $\Aa^{-1}$ is ensured by theorem \ref{thm:FIconvergence}.
    \end{proof}
    
    \begin{remark}
      A more general version of this theorem can be found in \cite[thm. 5.6]{Berkesch.Forsgard.PassareEulerMellinIntegrals2011} with an independent proof. In \cite{SymanzikCalculationsConformalInvariant1972} a similar technique is used to obtain Mellin-Barnes representations from Feynman integrals.
      
      The convergence of those Mellin-Barnes representation is discussed in \cite{Hai.SrivastavaConvergenceProblemCertain1995} and general aspects of multivariate Mellin-Barnes integrals can be found in \cite{Paris.KaminskiAsymptoticsMellinBarnesIntegrals2001}.
      
      The corollary can alternatively be proven by splitting the term $G^{-\nu_0}$ by the multinomial theorem in a multidimensional power series and solve the Mellin transform of this power series by a generalized version of Ramanujan's master theorem \cite{Gonzalez.Moll.SchmidtGeneralizedRamanujanMaster2011}.
    \end{remark}

    The type of the integral which appears in equation (\ref{eq:MBFox}) is also known as multivariate Fox $H$-function \cite{Hai.SrivastavaConvergenceProblemCertain1995} and the connection between Feynman integrals and Fox $H$-function was studied before \cite{Inayat-HusNewPropertiesHypergeometric1987, Inayat-HusNewPropertiesHypergeometric1987a, Buschman.SrivastavaFunctionAssociatedCertain1990}.
    
    Since the number of monomials in the Symanzik polynomials increases fast for more complex Feynman graphs, the Mellin-Barnes representation of theorem \ref{thm:MB} does not provide an efficient way to calculate those integrals. One exception are the massless $2$-point functions consisting in $n$ lines and having the loop number $L=n-1$. These so-called ``banana graphs'' or ``sunset-like'' graphs, have only $n+1$ monomials in $G$ and satisfy therefore the condition in the corollary. 
    
    In order to illustrate the theorem \ref{thm:MB} about Mellin-Barnes representations, we finish this section with a simple example.
    
    \begin{figure}[ht] 
      \begin{center}
        \vspace{-1cm}
        \includegraphics[width=.38\textwidth]{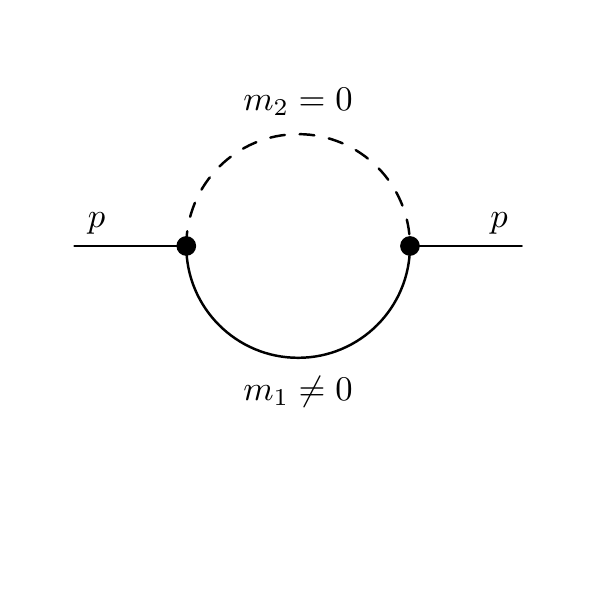}
        \vspace{-2cm}
      \end{center}
      \caption{The $1$-loop $2$-point function with one mass.}
      \label{fig:bubble1}
    \end{figure}

    \begin{example} \label{ex:1loopbubbleA}
      Consider the $1$-loop $2$-point function with one mass (see figure \ref{fig:bubble1}) having the Symanzik polynomials $U=x_1 + x_2$ and $F=(m_1^2-p^2) x_1 x_2 + m_1^2 x_1^2$. Thus the matrix $\Aa$ and the vector $z$ are given by
      \begin{align}
      \Aa = \begin{pmatrix} 
              1 & 1 & 1 & 1 \\
              1 & 0 & 1 & 2 \\
              0 & 1 & 1 & 0
            \end{pmatrix} \qquad z = (1,1,m_1^2-p^2,m_1^2) \point
      \end{align}
      Choosing $\sigma = \{1,2,3\}$, the corresponding Feynman integral in the Mellin-Barnes representation of theorem \ref{thm:MB} is given by
      \begin{align}
        J_\Aa (\nuu,z) &= z_1^{-\nu_0+\nu_2} z_2^{-\nu_0+\nu_1} z_3^{\nu_0-\nu_1-\nu_2} \int_{\delta-i\infty}^{\delta+i\infty} \frac{\d[t]}{2\pi i} \Gamma(t) \Gamma(\nu_0-\nu_1+t)\nonumber\\
        &\qquad \qquad \Gamma(\nu_0-\nu_2-t) \Gamma(-\nu_0+\nu_1+\nu_2-t) \left(\frac{z_2 z_4}{z_1 z_3}\right)^{-t} \point
      \end{align}
      For the correct contour prescription the poles have to be separated such that there exist values $\delta$ satisfying $\max \{0,-\nu_0+\nu_1\} < \delta < \min \{\nu_0-\nu_2,-\nu_0+\nu_1+\nu_2\}$. In the formulation to be introduced in the following section \ref{sec:FIprop}, this is equivalent to claim a full dimensional Newton polytope of $G$. In this case, by Cauchy's theorem the integral evaluates simply to a Gaussian hypergeometric function
      \begin{align}
        J_\Aa (\nuu,z) &=  \frac{\Gamma(2\nu_0-\nu_1-\nu_2)\Gamma(\nu_2)\Gamma(\nu_0-\nu_2)\Gamma(-\nu_0+\nu_1+\nu_2)}{\Gamma(\nu_0)} \nonumber\\
       & z_1^{-\nu_0+\nu_2} z_2^{-\nu_0+\nu_1} z_3^{\nu_0-\nu_1-\nu_2} \HypF{\nu_0-\nu_2,-\nu_0+\nu_1+\nu_2}{\nu_0}{1-\frac{z_2 z_4}{z_1z_3}} \point
      \end{align}
      After restoring the original prefactors and coefficients $z_1=z_2=1$, $z_3=p^2-m_1^2$, $z_4=m_1^2$ and $\nu_0=\frac{d}{2}$ it agrees with the expected result
      \begin{align}
        I_\Gamma (\nu_1,\nu_2,d,m_1^2,p^2) &= \frac{\Gamma\left(\frac{d}{2}-\nu_2\right)\Gamma\left(-\frac{d}{2}+\nu_1+\nu_2\right)}{\Gamma(\nu_1)\Gamma\left(\frac{d}{2}\right)} \nonumber\\
        & (p^2-m_1^2)^{\frac{d}{2}-\nu_1-\nu_2} \HypF{\frac{d}{2}-\nu_2,-\frac{d}{2}+\nu_1+\nu_2}{\frac{d}{2}}{1-\frac{m_1^2}{p^2-m_1^2}} \point
      \end{align}
    \end{example}


    \subsection{Mathematical Interlude: Convex Polytopes} \label{ssec:interlude}
    
    In the following, it will be fruitful to describe properties and solutions of Feynman integrals from a perspective of convex polytopes. In order to clarify some terminology, we give a short review of the basic concepts of convex polytopes. Readers which are familiar with convex polytopes can skip this interlude. More details about polytopes, customized to the hypergeometric approach, can be found in the appendix. For further treatments we refer to \cite{BrondstedIntroductionConvexPolytopes1983, Henk.Richter-Ge.ZieglerBasicPropertiesConvex2004b, BrunsPolytopesRingsKTheory2009, DeLoera.Rambau.SantosTriangulationsStructuresAlgorithms2010}.\spacepar
    
    A \textit{convex polytope} $P$ is defined as the convex hull of a finite set of points $\mathrm{A} = \{\mathrm{a}_1,\ldots,\mathrm{a}_N\}$
    \begin{align}
      P := \Conv(\mathrm A) := \left\{ \sum_{j=1}^N k_j \mathrm{a}_j \, \left| \, k \in \mathbb R_{\geq 0}^N, \sum_{j=1}^N k_j = 1 \right\} \right.
    \end{align}
    where $\mathrm a_i \in \mathbb R^n$. Additionally, if the points $\mathrm a_i \in \mathbb Z^n$ form an integer lattice, $P$ is called a convex lattice polytope. Since in this discussion all polytopes will be convex polytopes, we will call them simply ``polytopes'' in the following.
    
    As a fundamental result of polytope theory, every polytope can also be written as a bounded intersection of half-spaces
    \begin{align}
      P := P(M,b) := \{ \mu \in\mathbb{R}^n |  m_j^T \cdot \mu \leq b_j, 1\leq j \leq k \} \label{eq:Hpolytope}
    \end{align}
    where $b\in \mathbb R^k$ and $M\in\mathbb R^{k\times n}$ are real, $m_j^T$ denotes the rows of the matrix $M$ and $\cdot$ is the standard scalar product.
    
    The \textit{Newton polytope} $\Delta_f$ corresponding to a multivariate polynomial $f=\sum_{\alpha} c_\alpha x^\alpha \in\mathbb C[x_1,\ldots,x_n]$ is defined as the convex hull of its exponent vectors $\Delta_f := \Conv(\{\alpha | c_\alpha\neq 0\})$.
    
    A subset of $P$ having the form $F:= \{\mu\in P | c \, \mu = \beta\}$, with $c \in \mathbb R ^{k\times n}$ and where every point of the polytope $\mu\in P$ satisfies the inequality $c \, \mu\leq \beta$, is called a \textit{face} of $P$. Faces of dimension zero are called \textit{vertices} and faces of dimension $n-1$ are called \textit{facets}. Thus, if equation (\ref{eq:Hpolytope}) consist in a minimal set of inequalities, the facets of $P$ are given by $m^T_j \mu = b_j$. The polytope without its faces is called the \textit{relative interior} $\relint P$ of the polytope $P$.

    The \textit{dimension} of a polytope $P= \Conv (\mathrm A)$ is defined as the dimension of its affine hull, and thus one can easily see
    \begin{align}
      \dim (\Conv (\mathrm A )) = \rank \begin{pmatrix} 1 \\ \mathrm A\end{pmatrix} - 1 \point
    \end{align}
    As before in equation (\ref{eq:Aa}), we will denote the matrix $\mathrm A$ with an additional row $(1,\ldots,1)$ as $\Aa = \begin{pmatrix} 1\\ \mathrm A\end{pmatrix}$. If a polytope $P\subset \mathbb R^n$ has the dimension $n$ it is called to be \textit{full dimensional} and \textit{degenerated} otherwise.
    
    For full dimensional polytopes it is meaningful to introduce a volume of polytopes. As the \textit{volume} $\vol_0 (P)\in\mathbb N$ of a lattice polytope we understand a volume, which is normalized such that the standard simplex has volume $1$. Therefore, this volume is connected to the standard Euclidean volume $\vol (P)$ by a factorial of the dimension  $\vol_0 (P) = n! \vol (P)$. 
    For a full dimensional simplex $P_\bigtriangleup = \Conv (\mathrm A)$ the volume is given by the determinant $\vol_0 (P_\bigtriangleup) = |\det \Aa|$. Thus, one way of calculating the volume of a polytope is by dividing the polytope in simplices. \spacepar
    
    A subdivision of a polytope $P$ into simplices, where the union of all simplices is the full polytope and the intersection of two distinct simplices is either empty or a proper face of both simplices, is called a \textit{triangulation}.
    
    A triangulation $T(\omega) = \{\sigma_1,\ldots,\sigma_r \}$ of a polytope $P=\Conv(\mathrm A)$ is called \textit{regular}, if there exists an height vector $\omega\in\mathbb R^N$, such that for every simplex $\sigma_i$ of this triangulation there exists another vector $r_i\in\mathbb R^{n+1}$ satisfying
    \begin{align*}
      r_i \cdot a_j &= \omega_j \qquad \textrm{for} \qquad j \in \sigma_i  \\
      r_i \cdot a_j &< \omega_j \qquad \textrm{for} \qquad j \notin \sigma_i \comma
    \end{align*}
    where $\cdot$ denotes the scalar product.
    
    It can be shown, that every convex polytope admits always a regular triangulation \cite{DeLoera.Rambau.SantosTriangulationsStructuresAlgorithms2010}. If all $\sigma\in T(\omega)$ belongs to simplices with volume $1$ (i.e. $|\det\Aas|=1$) the triangulation is called \textit{unimodular}.
    
    In order to reduce later more complicated Feynman integrals to simpler one, we introduce subtriangulations.
    \begin{defn}[Subtriangulation]
      Let $P\subset \mathbb R^n$ be a polytope and $T$ be a triangulation of $P$. A simplicial complex $T^\prime$ is a \textit{subtriangulation} of $T$ (symbolically $T^\prime\subseteq T$) if
      \begin{enumerate}[a)]
        \item $T^\prime$ is a subcomplex of $T$ (i.e. $T^\prime$ is a subset of $T$, and $T^\prime$ is a polyhedral complex) and
        \item the union of all simplices in $T^\prime$ equals a (convex) polytope $P^\prime \subseteq P$.
      \end{enumerate}
      If the triangulations differ $T^\prime \neq T$ the subtriangulation is a \textit{proper subtriangulation} $T^\prime \subset T$.
    \end{defn}

    \begin{example}
      Consider the point configuration from example \ref{ex:1loopbubbleA}
      \begin{align}
        \Aa = \begin{pmatrix}
              1 & 1 & 1 & 1\\
              1 & 0 & 1 & 2\\
              0 & 1 & 1 & 0
              \end{pmatrix} \point
      \end{align}
      There are two regular triangulations $T_1 = \{\{1,2,4\},\{2,3,4\}\}$ and $T_2 =  \{ \{1,2,3\},\{1,3,4\}\}$. The first can be obtained e.g. by a height vector $\omega_1=(0,0,1,0)$. For the other triangulation one can consider $\omega_2=(0,0,0,1)$.
      
     The possible proper subtriangulations of these two triangulations are simply the single simplices itself.
     
     \begin{figure}[ht]

      \begin{center}
        \begin{tabular}{ccc}
	  \begin{tikzpicture}[scale=1]
            \draw[step=1cm,gray,very thin] (-0.1,-0.1) grid (2.2,1.7);  
	    \draw[thick,->] (0,0) -- (2.2,0) node[anchor=north west] {$\mu_1$};
	    \draw[thick,->] (0,0) -- (0,1.7) node[anchor=south east] {$\mu_2$};
	    \coordinate[label=below:$1$] (A) at (1,0);
	    \coordinate[label=left:$2$] (B) at (0,1);
	    \coordinate[label=above:$3$] (C) at (1,1);
	    \coordinate[label=below:$4$] (D) at (2,0);
	    \fill (A) circle (1pt);
	    \fill (B) circle (1pt);
	    \fill (C) circle (1pt);
	    \fill (D) circle (1pt);
	    \draw[thick] (A) -- (B);
	    \draw[thick] (C) -- (D);
	    \draw[thick] (D) -- (A);
	    \draw[thick] (B) -- (C);
	    \draw[thick] (B) -- (D);
	  \end{tikzpicture}
	  &
	  \hspace{1cm}
	  &
	  \begin{tikzpicture}[scale=1]
	    \draw[step=1cm,gray,very thin] (-0.1,-0.1) grid (2.2,1.7);  
	    \draw[thick,->] (0,0) -- (2.2,0) node[anchor=north west] {$\mu_1$};
	    \draw[thick,->] (0,0) -- (0,1.7) node[anchor=south east] {$\mu_2$};
	    \coordinate[label=below:$1$] (A) at (1,0);
	    \coordinate[label=left:$2$] (B) at (0,1);
	    \coordinate[label=above:$3$] (C) at (1,1);
	    \coordinate[label=below:$4$] (D) at (2,0);
	    \fill (A) circle (1pt);
	    \fill (B) circle (1pt);
	    \fill (C) circle (1pt);
	    \fill (D) circle (1pt);
	    \draw[thick] (A) -- (B);
	    \draw[thick] (B) -- (C);
	    \draw[thick] (C) -- (D);
	    \draw[thick] (D) -- (A);
	    \draw[thick] (A) -- (C);
	  \end{tikzpicture} \\
	  The triangulation $T_1$ of $\Conv(\mathrm A)$ & & The triangulation $T_2$ of $\Conv(\mathrm A)$ \\
	  generated by $\omega_1 =(0,0,1,0)$ & & generated by $\omega_2 =(0,0,0,1)$
	\end{tabular}
      \end{center}
      \caption{The two possible regular triangulations of the Newton polytope $\Delta_G=\Conv(\mathrm A)$ corresponding to the Lee-Pomeransky polynomial $G= z_1 x_1+z_2x_2+z_3x_1 x_2+z_4x_1^2$.}
     \end{figure}
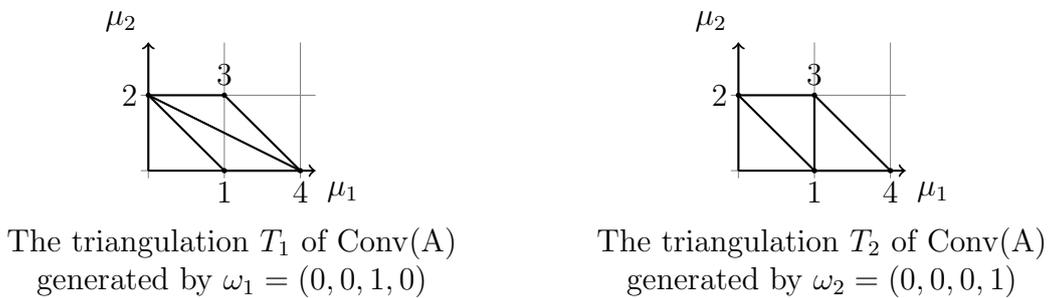

    \end{example}


    \subsection{Properties of Feynman Integrals} \label{sec:FIprop}
    In the definitions of the Feynman integrals we have omitted to discuss the convergence of those integrals. E.g. by power counting \cite{WeinbergHighEnergyBehaviorQuantum1960} the convergence behavior of the original Feynman integrals in momentum space is well known, as well as for the parametric Feynman integral (\ref{eq:FI1}) \cite{PanzerFeynmanIntegralsHyperlogarithms2015}. Therefore, we discuss the convergence shortly for the representation (\ref{eq:leepom}) and (\ref{eq:FI-J}), respectively. This discussion involves the perspective of polytopes, which allows a clear and short notation. The theorems are mostly direct implications of the work of \cite{Nilsson.PassareMellinTransformsMultivariate2010, Berkesch.Forsgard.PassareEulerMellinIntegrals2011, SchultkaToricGeometryRegularization2018} and proofs can be found there.
    \begin{theorem}[following from {\cite[thm. 2.2]{Berkesch.Forsgard.PassareEulerMellinIntegrals2011}}, the second statement is proven in {\cite[thm. 3.1]{SchultkaToricGeometryRegularization2018}}] \label{thm:FIconvergence}
      Consider the Feynman integral (\ref{eq:FI-J}) in the Euclidean region $\Re z_j>0$ with positive dimensions $\Re \nu_0 >0$ and the Lee-Pomeransky polynomial $G=U+F$. Denote by $\Delta_G$ the Newton polytope of $G$ and by $\relint \Delta_G$ its relative interior. Then the Feynman integral converges absolutely if the real parts of $\nu$ scaled componentwise by the real part of $\nu_0=\frac{d}{2}$ lie in the relative interior of the Newton polytope
      \begin{align}
        \Re(\nu)/\Re(\nu_0) \in \relint \Delta_G \point
      \end{align}
      In the description of (\ref{eq:Hpolytope}) this is equivalent to demanding $b_j \Re \nu_0 - m_j^T \cdot \Re \nu > 0$ for $1\leq j \leq k$. Furthermore, if the Newton polytope $\Delta_G$ is not full-dimensional, the Feynman integral does not converge absolutely for any choice of $\nu_0$ and $\nu$.
    \end{theorem}
    The second statement means, that Feynman integrals which do not have a full-dimensional Newton polytope $\Delta_G$ are neither dimensionally nor analytically regularizable. This result is not surprising, since if the Newton polytope is not full dimensional, the polynomial $G$ has a special homogeneous property. A polynomial $G$ is homogeneous in such that way, that there exists numbers $c_0, \ldots, c_n\in\mathbb Z$ not all zero, such that
    \begin{align}
      G_z (s^{c_1} x_1 ,\ldots,s^{c_n} x_n ) = s^{c_0} G_z (x_1,\ldots,x_n)
    \end{align}
    holds. If $G$ corresponds to a massless tadpole graph, $G$ is homogeneous in this sense, since massless tadpole graphs have $F=0$ and the first Symanzik polynomial is homogeneous of degree $L$. Thus, there will be no values for $\nuu=(\nu_0,\nu)\in\mathbb C^{n+1}$ that those integrals converge. However one usually sets those tadpole graphs to zero by including a counterterm in the renormalization procedure, which removes all tadpole graphs \cite{StermanIntroductionQuantumField1993} (see also \cite[sec. 3.1]{Bitoun.Bogner.Klausen.PanzerFeynmanIntegralRelations2019}).
    
    Since the Feynman integral in (\ref{eq:FI-J}) is only equivalent to the original Feynman integral in the meromorphic extension, the convergence region of (\ref{eq:FI-J}) has no deeper physical meaning. Using the following theorem one can deduce the meromorphic continuation of (\ref{eq:FI-J}) to the whole complex plane, which determines the physical relevant poles of the Feynman integral.
  
    \begin{theorem}[Meromorphic continuation of Feynman integrals {\cite[thm. 2.4, rem. 2.6]{Berkesch.Forsgard.PassareEulerMellinIntegrals2011}}] \label{thm:meromorphic}
      Consider a Feynman integral $J_\Aa$  in the Euclidean region $\Re z_j>0$ with a full-dimensional Newton polytope $\Delta_G = \{ \mu \in\mathbb{R}^n |  m_j^T \cdot \mu \leq b_j, 1\leq j \leq k \}$. Then one can rewrite the Feynman integral as
      \begin{align}
        J_\Aa (\nuu,z) = \Phi_\Aa(\nuu,z) \prod_{j=1}^{k} \Gamma(b_j \Re \nu_0 - m_j^T \cdot \Re \nu) \label{eq:mero}
      \end{align}
      where $\Phi_\Aa(\nuu,z)$ is an entire function with respect to $\nuu\in\mathbb C^{n+1}$.
    \end{theorem}
 
    \begin{example}
      Consider the example from above which correspond to figure \ref{fig:bubble1}. For the relative interior of the Newton polytope one obtains from the facet representation, the region of convergence (with $\Re\nu_0>0$)
      \begin{align*}
        & \qquad \Re(\nu_0-\nu_2) > 0 & \Re(-\nu_0+\nu_1+\nu_2)>0 \\
        & \qquad \Re(\nu_2) > 0 & \Re(2\nu_0-\nu_1-\nu_2) > 0
      \end{align*}
      which enables us to separate the poles of the Feynman integral in the $\Gamma$ functions
      \begin{align*}
        I_\Aa (\nuu,z) = \Phi_\Aa (\nuu,z) \frac{\Gamma(-\nu_0+\nu_1+\nu_2) \Gamma(\nu_0-\nu_2)}{\Gamma(\nu_1)}
      \end{align*}
      with an entire function $\Phi_\Aa(\nuu,z)$.
    \end{example}

    This result is remarkable in many different ways. Firstly, it guarantees that the Feynman integral can be meromorphically continued to the whole complex plane, which confirms dimensional and analytical regularization. Secondly, equation (\ref{eq:mero}) gives an easy method to calculate the (possible) poles of the Feynman integral. And thirdly, in the $\epsilon$ expansion one can focus only on a Taylor expansion of $\Phi_\Aa$ instead of a Laurent expansion of $J_\Aa$. Thus, one can determine the coefficients by differentiating, which makes the procedure much easier. \spacepar
  
    These theorems do not rely on any special properties of the Symanzik polynomials. In fact, Feynman integrals are just a subset of Euler-Mellin integrals. The following lemma is a simple implication from the properties of Symanzik polynomials to be at most quadratic, and to be homogeneous of degree $L$ and $L+1$, respectively.
    \begin{lemma} Let $\Aa = \begin{pmatrix}1\\ \mathrm A \end{pmatrix}$ be a point configuration coming from a Feynman graph.
      \begin{enumerate}[a)]
        \item The entries of the matrix $\Aa$ are restricted to $\Aa\in \{0;1;2\}^{(n+1)\times N}$. Every column of $\Aa$ contains at most one entry equals $2$. For massless Feynman integrals the points are even in $\Aa\in\{0;1\}^{(n+1)\times N}$.
        \item All points in $\mathrm A$ are arranged on two parallel hyperplanes in $\mathbb R^{n}$. The hyperplanes have the normal vector $(1,1,\ldots,1)$ and a distance of $1$ between them. This means that the Newton polytopes $\Delta_G$ arising in Feynman integrals are compressed in one direction.
        \item The Newton polytope $\Delta_G$, corresponding to a Feynman graph, has no interior points.
      \end{enumerate}
    \end{lemma}
    
    \vspace{1cm}
    
    Thus, we have reformulated the Feynman integral as an Euler-Mellin integral, which defines a meromorphic function in $\nuu\in\mathbb C^{n+1}$. To consider the Euclidean region we have restricted\footnote{In \cite{Berkesch.Forsgard.PassareEulerMellinIntegrals2011} there is a possibility to remove this limitation by considering the coamoeba of $G$.} the discussion to the right half space $z\in\mathbb C^N$ with $\Re z_j>0$. Further, in corollary \ref{cor:Fox} we had a class of Feynman integrals which provides a simple and analytic solution. These integrals will be helpful to find boundary values for the following partial differential equation systems. The Newton polytopes $\Delta_G$, which define the convergence regions and determine the whole structure of the Feynman graphs, which is necessary to evaluate Feynman amplitudes, are relatively well behaved.


  \section{General Hypergeometric Functions} \label{sec:GKZ}
  
  Since the first hypergeometric function was studied by Euler and Gauss more than 200 years ago, many different generalization of hypergeometric functions were introduced: Pochhammer series ${}_pF_q$, Appell's, Lauricella's and Kampé-de-Fériet functions, to name a few. Those functions can be characterized in three different ways: by series representations, by integral representations and as solutions of partial differential equations. Therefore, there are in principle three different branches to generalize the notion of a hypergeometric function.
  
  The most general series representation goes back to Horn \cite{HornUberHypergeometrischeFunktionen1940} and was later investigated by Ore and Sato (a summarizing discussion can be found in \cite{Gelfand.Graev.RetakhGeneralHypergeometricSystems1992}). A Horn hypergeometric series is a multivariate power series in the variables $x_1,\ldots,x_r\in\mathbb C$
  \begin{align}
    \sum_{k\in\mathbb N_0^r} c(k) x^k
  \end{align}
  where ratios of coefficients $\frac{c(k+e_i)}{c(k)}$ are rational functions in $k_1,\ldots k_r$ and where $e_i$ is the standard basis in the Euclidean space. Thus, the coefficients can be represented mainly by a product of Pochhammer symbols\footnote{By the property of the Pochhammer symbols to satisfy $(a)_n^{-1} = (1-a)_{-n}$ for $n\in\mathbb Z$ one can convert Pochhammer symbols in the denominator to Pochhammer symbols in the numerator and vice versa. The most general form of those terms are given by the Ore-Sato theorem \cite{Gelfand.Graev.RetakhGeneralHypergeometricSystems1992}.}, which are defined as
  \begin{align}
    (a)_{n} := \frac{\Gamma(a+n)}{\Gamma(a)} \point
  \end{align}
  Negative integers in arguments of $\Gamma$ functions can be avoided by considering appropriate limits. In the terms of the series of Horn hypergeometric functions we consider $a\in\mathbb C$ as complex numbers and $n$ as integer combinations of the $k_1,\ldots,k_r$. Among many beautiful properties, derivatives with respect to the parameters of Horn hypergeometric are again Horn hypergeometric functions \cite{Bytev.Kniehl.MochDerivativesHorntypeHypergeometric2017}. For further studies of Horn hypergeometric functions we refer e.g. to \cite{SadykovHypergeometricFunctionsSeveral2002}.

    \subsection{GKZ Hypergeometric Functions} \label{ssec:GKZHyp}
    Since the late 1980s the theory of general hypergeometric functions was reinvented by the characterization of the partial differential equation system by Gelfand, Graev, Kapranov, Zelevinsky and collaborators \cite{GelfandPartGeneralTheory1989, Gelfand.Graev.ZelevinskiHolonomicSystemsEquations1987, Gelfand.Kapranov.ZelevinskiHypergeometricFunctionsToric1991, Gelfand.Kapranov.ZelevinskyGeneralizedEulerIntegrals1990, Gelfand.Graev.RetakhGeneralHypergeometricSystems1992, Gelfand.Kapranov.ZelevinskyDiscriminantsResultantsMultidimensional2008a}. This approach combines the different characterizations of hypergeometric functions and gives a comprehensive method to analyze and describe general hypergeometric functions. In this section we introduce the Gelfand-Kapranov-Zelevinsky (GKZ) hypergeometric system, which is a system of partial differential equations and discuss roughly how to solve this system by power series. The following recapitulation is adapted for the application to Feynman integrals and will miss some generality in order to keep the discussion short. Since the theory of general hypergeometric functions involves many mathematical aspects inter alia algebraic geometry, combinatorics, number theory and Hodge theory, we refer for more detailed studies to \cite{GelfandPartGeneralTheory1989, Gelfand.Graev.RetakhGeneralHypergeometricSystems1992, Aomoto.Kita.KohnoTheoryHypergeometricFunctions2011, Saito.Sturmfels.TakayamaGrobnerDeformationsHypergeometric2000, StienstraGKZHypergeometricStructures2005, CattaniThreeLecturesHypergeometric2006, Vilenkin.KlimykGelFandHypergeometric1995}. \spacepar
    
    Let $\Aa\in\mathbb{Z}^{(n+1)\times N}$ be an integer matrix with $n+1\leq N$, $\rank \Aa = n+1$ and assume that the columns of $\Aa$ span the full integer lattice\footnote{The existence of a subset $\Aa^\prime\subseteq \Aa$, which forms a full dimensional, unimodular simplex $\Conv \Aa^\prime$, is a sufficient condition for this assumption.} $\operatorname{colspan}_{\mathbb Z} \Aa = \mathbb Z^{n+1}$. Further, we consider that $\Aa$ includes a row\footnote{One can generalize this condition to $\operatorname{rowspan}_{\mathbb Q} \Aa = (1,\ldots,1)$. However, for the following this generalization is not necessary.} of the form $(1,1,\ldots,1)$. The ladder condition means that $\Aa$ lies in an $n$-dimensional affine hyperplane of $\mathbb Z^{n+1}$. Without loss of generality we can consider that $\Aa$ is of the form $\Aa = \begin{pmatrix} 1 \\ \mathrm A \end{pmatrix}$ as in equation (\ref{eq:Aa}). Then the Gelfand-Kapranov-Zelevinsky (GKZ) hypergeometric system is defined as the $D$-module, consisting of toric and homogeneous differential operators
    \begin{align}
      H_\Aa (\beta) = \{\partial^u - \partial^v | \Aa u = \Aa v, u,v\in\mathbb N^N\} \cup \langle \Aa \theta + \beta \rangle \label{eq:GKZ}
    \end{align}
    where $\theta = (z_1 \frac{\partial}{\partial z_1},\ldots,z_N \frac{\partial}{\partial z_N})$ is the Euler operator and $\beta\in\mathbb C^{n+1}$. By $\langle D \rangle$ we denote the ideal generated by components of $D$ over $\mathbb C$. Solutions $\Phi(z)$ of these differential systems $H_\Aa (\beta) \Phi(z)=0$ are called $\Aa$-hypergeometric functions.
    
    One of the most important properties of those hypergeometric systems (\ref{eq:GKZ}) is to be holonomic, i.e. the dimension of the solution space is finite and we can give an appropriate basis of the solution space. There are several possibilities to construct those bases. In the following we discuss a solution in terms of multivariate power series, which was the first solution invented in \cite{Gelfand.Graev.ZelevinskiHolonomicSystemsEquations1987}. Furthermore, there are solutions in terms of different sorts of integrals \cite{Matsubara-LaplaceResidueEuler2018}. The following discussion is simplified for the later approach and based on \cite{Fernandez-IrregularHypergeometricDmodules2009} and \cite{StienstraGKZHypergeometricStructures2005}.
  
    Let $\Aa\in\mathbb Z^{(n+1)\times N}$ be an integer matrix with $n+1\leq N$ and full rank as before. Furthermore consider the corresponding lattice $\mathbb {L} := \operatorname{ker}_\mathbb{Z} \Aa = \{ (l_1,\ldots,l_N)\in\mathbb{Z}^N | l_1 a_1 + \ldots + l_N a_N = 0\}$, which has $\rank \mathbb L = N-n-1 =: r$ by the rank-nullity theorem. Then for $\xi\in\mathbb{C}^N$ the formal series
    \begin{align}
      \varphi_\xi(z) = \sum_{l\in\mathbb L} \frac{z^{l+\xi}}{\Gamma(\xi+l+1)}
    \end{align}
    is called \textit{$\mathit\Gamma$-series}. It turns out, that these $\Gamma$-series are formal solutions of the GKZ system (\ref{eq:GKZ}) for an appropriate choice of $\xi$:
    
    \begin{lemma}[$\Gamma$-series as formal solutions of GKZ hypergeometric systems \cite{Gelfand.Graev.ZelevinskiHolonomicSystemsEquations1987, StienstraGKZHypergeometricStructures2005}]
      Let $\mathbb{L}$ be the corresponding lattice to $\Aa$ and $\xi\in\mathbb C^N$ satisfying $\Aa \xi + \beta = 0$. Then the series $\varphi_\xi(z)$ is a formal solution of the GKZ system $H_\Aa (\beta)$
      \begin{align*}
        H_\Aa (\beta) \varphi_\xi(z)=0 \point
      \end{align*}
    \end{lemma}
    \begin{proof}
      For $u\in\mathbb N^N$ and $r\in\mathbb C^N$ it is $\left(\frac{\partial}{\partial z}\right)^u z^r = \frac{\Gamma(r+1)}{\Gamma(r-u+1)} z^{r-u}$ (with an appropriate limit, respectively). Furthermore one can add an element of $\mathbb L$ to $\xi$, without changing the $\Gamma$-series. Since $u-v\in\mathbb L$ it is
      \begin{align}
        \partial^u \varphi_{\xi}(z) &= \sum_{l\in\mathbb L} \frac{z^{l+\xi-u}}{\Gamma(\xi+l-u +1)} = \varphi_{\xi-u}(z) = \varphi_{\xi-v}(z) = \partial^v\varphi_{\xi}(z)
      \end{align}
      which shows that the $\Gamma$-series satisfies the toric equations. For the homogeneous equations one considers
      \begin{align}
        \sum_{j=1}^N a_j z_j \frac{\partial}{\partial z_j} \varphi_\xi(z) &= \sum_{l\in\mathbb L} \left(\sum_{j=1}^N a_j (\xi_j+l_j)\right) \frac{z^{l+\xi}}{\Gamma(\xi+l+1)} \nonumber \\
        &= \sum_{j=1}^N a_j \xi_j \sum_{l\in\mathbb L}  \frac{z^{l+\xi}}{\Gamma(\xi+l+1)} = - \beta \varphi_\xi(z) \point
      \end{align}
    \end{proof}
   
    The restriction $\Aa \xi + \beta=0$ allows in general many choices of $\xi$. Let $\sigma \subseteq \{1,\ldots,N\}$ be an index set with cardinality $n+1$, such that the matrix $\Aa$ restricted to columns of that index set $\sigma$ is invertible, $\det \Aas \neq 0$. Due to the assumption $\rank \Aa=n+1$ those index sets always exist. Denote by $\bar\sigma = \{1,\ldots,N\} \setminus \sigma$ the complement of $\sigma$. If one sets $\xi_\sigma = -\Aas^{-1} (\beta + \Aabs k)$ and $\xi_{\bar\sigma} = k$ the condition $\Aa \xi + \beta = 0$ is satisfied for any $k\in\mathbb{C}^{r}$.
    
    On the other hand we can split the lattice $\mathbb L = \{ l\in\mathbb Z^N | \Aa l =0\}$ in the same way $\Aas l_\sigma + \Aabs l_{\bar\sigma}=0 $ and obtain a series only over $l_{\bar\sigma}$
    \begin{align}
      \varphi_\xi (z) &= \sum_{\substack{l_{\bar\sigma}\in\mathbb Z^r \\ \textrm{s.t. } \Aas^{-1}\Aabs l_{\bar\sigma}\in\mathbb Z^{n+1} }} \frac{z_\sigma^{-\Aas^{-1} (\beta+\Aabs k + \Aabs l_{\bar\sigma}) } z_{\bar\sigma}^{k+l_{\bar\sigma}}}{\Gamma(-\Aas^{-1}(\beta+\Aabs k +\Aabs l_{\bar\sigma})+1) \Gamma(k+l_{\bar\sigma}+1)} \point
    \end{align}
    In order to simplify the series one can choose\footnote{Also the choice $k\in\mathbb Z^r$ would be possible, but it does not change the series, see \cite[lemma 3.2.]{Fernandez-IrregularHypergeometricDmodules2009}} $k\in\mathbb N_0^r$, since terms with $(k+l_{\bar\sigma})_i\in\mathbb{Z}_{< 0}$ will vanish. The $\Gamma$-series depends now on $k$ and $\sigma$
    \begin{align}
      \varphi_{\sigma,k}(z) =z_\sigma^{-\Aas^{-1}\beta} \sum_{\lambda\in\Lambda_k} \frac{z_\sigma^{-\Aas^{-1}\Aabs\lambda} z_{\bar\sigma}^\lambda}{\lambda! \Gamma(-\Aas^{-1}(\beta+\Aabs \lambda)+1)} \label{eq:GammaPowerSeriesVarphi}
    \end{align}
    where $\Lambda_k=\{ k + l_{\bar\sigma} \in \mathbb{N}_0^r | \Aabs l_{\bar\sigma} \in \mathbb Z \Aas \}\subseteq \mathbb{N}_0^r$ for any $k\in\mathbb N_0^r$. Therefore, the $\Gamma$-series is turned into a power series.
   
    \begin{remark}
      In the unimodular case $|\det \Aas|=1$, the coefficient matrix $\Aas^{-1}\Aabs\in\mathbb Z^{(n+1)\times r}$ is an integer matrix and therefore it is $\Lambda_k = \mathbb N^r_0$. Furthermore, the set $\{\Lambda_k | k\in\mathbb N^r_0 \}$ is a partition of $\mathbb N^r_0$ with cardinality $|\det \Aas|$ \cite{Fernandez-IrregularHypergeometricDmodules2009}.      
    \end{remark}
   
    In order to show that $\Gamma$-series are actual solutions of the GKZ system and not only formal ones, one has to prove that $\Gamma$-series converge for some $z\in\mathbb C^N$. By an application of the Stirling approximation it can be shown, that the $\Gamma$-series always converge absolutely for sufficiently small values of the variables $x_j := \frac{(z_{\bar\sigma})_j}{\prod_i (z_\sigma)_i^{(\Aas^{-1}\Aabs)_{ij}}}$. A proof of the absolute convergence of $\Gamma$-series can be found in lemma \ref{ssec:app-convergence}.
    
    Another issue is also that the $\Gamma$-series can be identical to zero, which is also inconvenient in order to construct a solution space. The $\Gamma$-series is zero for all $z\in\mathbb C^N$, if and only if for all $\lambda\in\Lambda_k$ the expression $\Aas^{-1} (\beta + \Aabs\lambda)$ contains at least one positive integer entry. To avoid these cases one considers generic $\beta\in\mathbb C^{n+1}$:
  
    \begin{defn}[Very genericity]
      If no component of $\Aas^{-1} (\beta+\Aabs\lambda)$ is a strictly positive integer for all $\lambda\in\mathbb{N}_{0}^r$ one says that $\beta$ is very generic with respect to $\sigma$. In the unimodular case this is equivalent to claim, that for the components $i$ which satisfy $(\Aas^{-1}\Aabs)_{ij}\geq 0$ for all $j$, it is $(\Aas^{-1} \beta)_i\notin \mathbb Z_{>0}$.
    \end{defn}
  
    Thus, typically non generic cases arise for even integer dimensions, which should not be surprising, since for these dimensions the Feynman integrals mostly diverge, which is the reason why we exclude this values already.
  
    To normalize the first term of the power series to $1$, we will deal in the following with a slightly different version of the $\Gamma$-series
    \begin{align}
      \phi_{\sigma,k} (z) := \Gamma(-\Aas^{-1}\beta+1) \varphi_{\sigma,k}(z) = z_\sigma^{-\Aas^{-1}\beta} \sum_{\lambda\in\Lambda_k} \frac{z_\sigma^{-\Aas^{-1}\Aabs\lambda} z_{\bar\sigma}^\lambda}{\lambda! (1 -\Aas^{-1}\beta)_{-\Aas^{-1}\Aabs \lambda}}
    \end{align}
    which is well-defined in the case $(\Aas^{-1}\beta)_i \notin \mathbb Z_{>0}$. Here $(1 -\Aas^{-1}\beta)_{-\Aas^{-1}\Aabs \lambda}$ denotes the (multivariate) Pochhammer symbol $(a)_n := \prod_j \frac{\Gamma(a_j+n_j)}{\Gamma(a_j)}$.
    
    \begin{remark}
      In the unimodular case $|\det \Aas |= 1 $ one can rewrite the $\Gamma$-series
      \begin{align}
        \phi_{\sigma} (z) = z_\sigma^{-\Aas^{-1}\beta} \sum_{\lambda\in\mathbb N_0^r} \frac{(\Aas^{-1}\beta)_{\Aas^{-1}\Aabs\lambda}}{\lambda!} \frac{z_{\bar\sigma}^\lambda}{(-z_\sigma)^{\Aas^{-1}\Aabs\lambda} } 
      \end{align}
      by Pochhammer identities.
    \end{remark}
 
    As mentioned above the holonomic rank is finite. For very generic $\beta$ one can determine the holonomic rank by a polytope corresponding to $\Aa=\begin{pmatrix} 1 \\ \mathrm A \end{pmatrix}$.
    
    \begin{theorem}[Holonomic rank of GKZ systems \cite{Gelfand.Kapranov.ZelevinskyGeneralizedEulerIntegrals1990, Gelfand.Kapranov.ZelevinskiHypergeometricFunctionsToric1991, CattaniThreeLecturesHypergeometric2006}] \label{thm:holrank}
      Consider a GKZ system $H_\Aa (\beta)$ with arbitrary $\Aa\in\mathbb Z^{(n+1)\times N}$ and very generic $\beta\in\mathbb C^{n+1}$. Let $\Conv (\mathrm A)$ be the corresponding convex polytope and denote by $\vol_0$ the normalized Euclidean volume, such that the standard simplex has a volume equal to $1$. Then the holonomic rank of the GKZ system is equal to the volume of the polytope $\Conv (\mathrm A)$
      \begin{align}
        \rank H_\Aa (\beta) = \vol_0 (\Conv (\mathrm A)) \point
      \end{align}
    \end{theorem} 
    
    This means that one needs $\vol_0 (\Conv (\mathrm A))$ linearly independent solutions to construct the solution space. The regular triangulations of the polytope $\Conv (\mathrm A)$ provide a construction of linearly independent $\Gamma$-series. Definitions and basic properties of the regular triangulations can be found in the mathematical interlude of section \ref{ssec:interlude} and in the appendix.
    
    In the following we will only discuss the case of unimodular triangulation, since almost all Feynman integrals admit an unimodular triangulation. We will motivate this restriction in section \ref{ssec:nonunimodular} in more detail. Nevertheless, there is a simple generalization to the non-unimodular case, which can be found e.g. in \cite{Fernandez-IrregularHypergeometricDmodules2009}, \cite{Matsubara-LaplaceResidueEuler2018}.
    
    \begin{theorem}[Solution Space of GKZ \cite{Gelfand.Kapranov.ZelevinskiHypergeometricFunctionsToric1991, Fernandez-IrregularHypergeometricDmodules2009}] \label{thm:solutionspace}
      Let $T$ be a regular unimodular triangulation and let $\beta$ be very generic with respect to every $\sigma\in T$. Then the $\Gamma$-series $\{\varphi_\sigma\}_{\sigma\in T}$ form a basis of the solution space of the hypergeometric GKZ system $H_\Aa (\beta)$. Furthermore, all these $\Gamma$-series have a common region of convergence.
    \end{theorem}
    
    \vspace{1cm}
    
    To conclude, we defined holonomic systems of partial differential equations, which can be characterized by a matrix $\Aa\in\mathbb Z^{(n+1)\times N}$ and a vector $\beta\in\mathbb C^{n+1}$. Furthermore, for generic values of $\beta$ we are able to construct the whole solution space in terms of power series by regular triangulations of the polytope $\Conv(\mathrm A)$.

    
    \section{Feynman Integrals as Hypergeometric Functions} \label{sec:FeynHyp}

    
    It is one of the first observations in the calculation of simple Feynman amplitudes, that Feynman integrals evaluate mostly to hypergeometric functions. This observation was leading Regge to the conjecture that Feynman integrals are always hypergeometric functions and he based his conjecture on the partial differential equations which are satisfied by the Feynman integral \cite{ReggeAlgebraicTopologyMethods1967}. 
    
    Typically, those hypergeometric functions also appear in the often used Mellin-Barnes approach. This is a consequence of Mellin-Barnes representations with integrands consisting in a product of $\Gamma$ functions, which can be identified by the hypergeometric Fox $H$-functions \cite{Inayat-HusNewPropertiesHypergeometric1987, Inayat-HusNewPropertiesHypergeometric1987a, Buschman.SrivastavaFunctionAssociatedCertain1990} or which can be evaluated to some series-based hypergeometric functions, like the Appell or Lauricella functions by application of the residue theorem \cite{Fleischer.Jegerlehne.TarasovNewHypergeometricRepresentation2003, Boos.DavydychevMethodEvaluatingMassive1991}. 
    But except of some special cases, like the one-loop integrals \cite{Fleischer.Jegerlehne.TarasovNewHypergeometricRepresentation2003}, this correspondence is more or less unproved, which is also due to the fact that multivariate Mellin-Barnes integrals can be highly non trivial \cite{Paris.KaminskiAsymptoticsMellinBarnesIntegrals2001}.
    
    A new opportunity to examine the correspondence between hypergeometric functions and Feynman integrals is the Gelfand-Kapranov-Zelevinsky approach. It was already stated by Gelfand himself, that ``practically all integrals which arise in quantum field theory'' \cite{Gelfand.Kapranov.ZelevinskyGeneralizedEulerIntegrals1990} can be treated with this approach. Recently, the connection between Feynman integrals and $\Aa$-hypergeometric function was concertized in \cite{NasrollahpPeriodsFeynmanDiagrams2016} and, independently of the present paper, in \cite{delaCruzFeynmanIntegralsAhypergeometric2019}.
    
    Based on the Lee-Pomeransky representation of Feynman integrals it is a standard application to show that generalized Feynman integrals are $\Aa$-hypergeometric, analogue to the examples in \cite{StienstraGKZHypergeometricStructures2005}.
   
    \begin{theorem}[Feynman integral as $\Aa$-hypergeometric function\protect\footnote{In \cite{NasrollahpPeriodsFeynmanDiagrams2016} it was proven, that the Feynman integrals satisfy a system of differential equations which is isomorphic to the GKZ system. Recently, the theorem was independently proven in \cite{delaCruzFeynmanIntegralsAhypergeometric2019} in a similar way.}] \label{thm:FeynmanAa}
      A generalized Feynman integral $J_\Aa (\nuu,z)$ satisfies the hypergeometric GKZ system in the variables $z\in\mathbb C^N$ 
      \begin{align}
        H_\Aa (\nuu) J_\Aa(\nuu,z)=0 \point      
      \end{align}
      Thus the generalized Feynman integral is an $\Aa$-hypergeometric function.
    \end{theorem}
    \begin{proof}
      Firstly, we show that the generalized Feynman integrals satisfy the toric part $ \{\partial^u - \partial^v | \Aa u = \Aa v, u,v\in\mathbb N^N\} $. Derivatives of the Feynman integral with respect to $z$ result in
      \begin{align}
        \partial^u \int_{\mathbb R^{n}_+} \d[x]x^{\nu-1}G_z(x)^{-\nu_0} = - \nu_0 (-\nu_0-1) \cdots (-\nu_0-|u| -1) \int_{\mathbb{R}^n_+} \d[x] x^{\nu-1} x^{\sum_{i} u_i \mathrm{a}_i} G_z(x)^{-\nu_0- |u|}
      \end{align}
      where $|u|:= \sum_i u_i$. From the row $(1,1,\ldots,1)$ in $\Aa$ it follows immediately that $|u|=|v|$. Therefore, one obtains the same equation for $v$.
      
      Secondly, consider for the homogeneous part $\langle \Aa \theta + \nuu \rangle$ that $J_\Aa(\nuu,s^{\mathrm{a}_{b1}} z_1,\ldots,s^{\mathrm{a}_{bN}} z_N) = \Gamma(\nu_0) \int_{\mathbb R^n_+} \d[x] x^{\nu-1} G_z(x_1,\ldots,s x_b,\ldots,x_n)^{-\nu_0}$. After a substitution $sx_b \rightarrow x_b$ for $s>0$ it is
      \begin{align}
        J_\Aa(\nuu,s^{\mathrm{a}_{b1}} z_1,\ldots,s^{\mathrm{a}_{bN}} z_N) = s^{-\nu_b} J_\Aa (\nuu,z) \point
      \end{align}
      A derivative with respect to $s$ completes the proof for $s=1$.
    \end{proof}
    
    Thus, as suggested already by Gelfand and confirmed in \cite{NasrollahpPeriodsFeynmanDiagrams2016} and \cite{delaCruzFeynmanIntegralsAhypergeometric2019}, every generalized Feynman integral with an Euclidean region satisfies the GKZ hypergeometric system and can be treated within the framework of GKZ. This will allow inter alia a series representation of the Feynman integral.

    \subsection{Hypergeometric Series Representations of Generalized Feynman Integrals} \label{ssec:series}
    
    As stated in theorem \ref{thm:FeynmanAa} generalized Feynman integrals are $\Aa$-hypergeometric functions. Thus, one can directly apply the results of hypergeometric GKZ systems from section \ref{sec:GKZ}. Consider first the case $\nuu\in (\mathbb C \setminus \mathbb Z)^{n+1}$ in order to satisfy the very genericity of $\nuu$ and to be sure that the Feynman integral has no poles. Later we can relax this strict claim.  
    
    We fix a regular, unimodular triangulation $T$ of the Newton polytope $\Delta_G$ of the Lee-Pomeransky polynomial $G=U+F$. Then by theorem \ref{thm:solutionspace}, we can write the generalized Feynman integral as a linear combination of $\Gamma$-series
    \begin{align}
      J_\Aa (\nuu,z) = \sum_{\sigma\in T} C_\sigma (\nuu) \phi_\sigma (\nuu,z) \label{eq:lincomb}
    \end{align}
    where the $z\in\mathbb C^N$ are defined in the positive half-space $\Re z_i \geq 0$ in the region where the $\Gamma$-series converge. Thus, one has to determine the meromorphic functions $C_\sigma (\nuu)$ in order to get a series representation of Feynman integrals. This can be done by comparing equation (\ref{eq:lincomb}) with boundary values of the Feynman integral. As we will see in the following, the Feynman integrals transmit their functions $C_\sigma(\nuu)$ to simpler Feynman integrals. This enables us to reduce Feynman integrals to the case described in corollary \ref{cor:Fox} and derive an analytic expression of the functions $C_\sigma(\nuu)$. 
    
    For this purpose, consider a generalized Feynman integral $J_\Aa $ with a Newton polytope $P = \Delta_G = \Conv (\mathrm A)$, consisting of the vertices $\mathrm a_1,\ldots, \mathrm a_N$, and an unimodular triangulation $T=\{\sigma_1,\ldots,\sigma_r,\eta_1,\ldots,\eta_s\}$ of $P$. Furthermore, let $T^\prime=\{\sigma_1,\ldots,\sigma_r\}$ be a proper subtriangulation of $T$. Denote the vertices of the convex polytope $P^\prime$, which correspond to the subtriangulation $T^\prime$, by $\mathrm  a_1,\ldots, \mathrm a_M$ with $M<N$ and the corresponding Feynman integral\footnote{$J^\prime_{\Aa^\prime}$ is not necessarily a Feynman integral coming from an actual Feynman graph. It is sufficient, that $J^\prime_{\Aa^\prime}$ has the shape of an Euler-Mellin integral described in definition \ref{def:genericFI}. Thus, $G_z^\prime$ can be an arbitrary non-homogeneous polynomial.} by $J^\prime_{\Aa^\prime}$. Thus, in the second Feynman integral some monomials in $G_z$ are missing
    \begin{align}
      J^\prime_{\Aa^\prime} (\nuu, z_1,\ldots, z_M) = \lim_{z_{M+1},\ldots,z_N \rightarrow 0} J_\Aa(\nuu,z_1,\ldots,z_N) \point
    \end{align}
    Applying the results from the previous sections to both Feynman integrals independently one obtains on the one hand
    \begin{align}
      J_\Aa(\nuu,z_1,\ldots,z_N)=\sum_{i=1}^r C_{\sigma_i}(\nuu) \phi_{\sigma_i}(\nuu,z_1,\ldots,z_N) + \sum_{i=1}^s C_{\eta_i}(\nuu) \phi_{\eta_i}(\nuu,z_1,\ldots,z_N)
    \end{align}
    with the $\Gamma$-series
    \begin{align}
      \phi_{\sigma_i}(\nuu,z_1,\ldots,z_N) &= z_{\sigma_i}^{-\Aa_{\sigma_i}\nuu} \sum_{\lambda\in\mathbb{N}_0^{|\bar\sigma_i|}} \frac{z_{\bar\sigma_i}^\lambda z_{\sigma_i}^{-\Aa_{\sigma_i}\Aa_{\bar\sigma_i}\lambda} }{\lambda!(1-\Aa_{\sigma_i}\nuu)_{-\Aa_{\sigma_i}\Aa_{\bar\sigma_i}\lambda}} \qquad i=1,\ldots,r\\
      \phi_{\eta_i}(\nuu,z_1,\ldots,z_N) &= z_{\eta_i}^{-\Aa_{\eta_i}\nuu} \sum_{\lambda\in\mathbb{N}_0^{|\bar\eta_i|}} \frac{z_{\bar\eta_i}^\lambda z_{\eta_i}^{-\Aa_{\eta_i}\Aa_{\bar\eta_i}\lambda} }{\lambda!(1-\Aa_{\eta_i}\nuu)_{-\Aa_{\eta_i}\Aa_{\bar\eta_i}\lambda}} \qquad i=1,\ldots,s
    \end{align}
    and on the other hand
    \begin{align}
      J^\prime_{\Aa^\prime} (\nuu,z_1,\ldots,z_M)=\sum_{i=1}^r C^\prime_{\sigma_i}(\nuu) \phi^\prime_{\sigma_i}(\nuu,z_1,\ldots,z_M) 
    \end{align}
    with 
    \begin{align}
      \phi^\prime_{\sigma_i}(\nuu,z_1,\ldots,z_M) &= z_{\sigma_i}^{-\Aa^\prime_{\sigma_i}\nuu} \sum_{\lambda\in\mathbb{N}_0^{|\bar\sigma_i|}} \frac{z_{\bar\sigma_i}^\lambda z_{\sigma_i}^{-\Aa^\prime_{\sigma_i}\Aa^\prime_{\bar\sigma_i}\lambda} }{\lambda!(1-\Aa^\prime_{\sigma_i}\nuu)_{-\Aa^\prime_{\sigma_i}\Aa^\prime_{\bar\sigma_i}\lambda}} \qquad i=1,\ldots,r \point
    \end{align}
    Due to the construction it is $\mathrm a_{M+1},\ldots, \mathrm a_N \in \bigcup_{i=1}^r \bar\sigma_i$. Therefore, in the limit $z_{M+1},\ldots,z_N\rightarrow 0$ only some $z_{\bar\sigma}$ will be affected and it is simply
    \begin{align}
      \lim_{z_{M+1},\ldots,z_N \rightarrow 0} \phi_{\sigma_i} (\nuu,z_1,\ldots,z_N) = \phi^\prime_{\sigma_i} (\nuu,z_1,\ldots,z_M) \point 
    \end{align}
    
    Apart from that, as a consequence of the homogeneous differential equations $\langle \Aa \theta+\nuu\rangle$ all solutions of a GKZ hypergeometric system have to satisfy the scaling property  $\phi (\nuu,s z) = s^{-\nu_0} \phi (\nuu,z)$. It can be easily seen that $\Gamma$-series satisfy this property. In the power series part of a $\Gamma$-series the variables $z$ only appear as ratios, such that they are scaling invariant (in accordance with lemma \ref{lem:AA1}). However, the monomial $z_\eta^{-\Aa_\eta^{-1} \nuu}$ in front of the power series will give the scaling property since $s^{\sum_{i\in\eta} (-\Aa_\eta^{-1} \nuu)_i} = s^{-\nu_0}$ according to lemma \ref{lem:AA1}.
    
    Since in every $\Gamma$-series $\phi_{\eta_i}$ some of the variables $z_{M+1}\ldots,z_N$ are contained in the monomial, the scaling property will be violated in the limit $z_{M+1}\ldots,z_N \rightarrow0$. Thus, the functions $\lim_{z_{M+1},\ldots z_N \rightarrow 0} \phi_{\eta_i}$ can not be linearly dependent of $\{\phi^\prime_{\sigma_i}\}$. Since one already has $r$ linearly independent solutions $\phi^\prime_{\sigma_i}$ and therewith has a full-dimensional solution space for the Feynman integral $J^\prime_{\Aa^\prime}$, the $\Gamma$-series $\phi_{\eta_i}$ have to vanish
    \begin{align}
      \lim_{z_{M+1},\ldots,z_N \rightarrow 0} \phi_{\eta_i} (\nuu,z_1,\ldots,z_N) = 0 \point
    \end{align}
    
    Thus, the GKZ systems behave naturally, as also mentioned in \cite{Gelfand.Graev.RetakhGeneralHypergeometricSystems1992}: If one deletes a vertex of $\Conv (\mathrm A)$, the $\Gamma$-series which correspond to simplices containing this vertex will vanish. This leads to a very simple connection between subtriangulations of Feynman integrals. Applying the limit $z_{M+1},\ldots, z_N \rightarrow 0$, the meromorphic functions $C_\sigma(\nuu)$ will not be affected and one obtains $C_{\sigma_i} (\nuu) = C^\prime_{\sigma_i} (\nuu)$.
    
    Thus, one can determine the meromorphic functions $C_\sigma (\nuu)$ by considering simpler Feynman integrals which refer to subtriangulations, where by a simpler Feynman integral we mean a Feynman integral where the Lee-Pomeransky polynomial $G$ has less monomials.  In this way one can define ancestors and descendants of Feynman integrals by deleting or adding monomials to the Lee-Pomeransky polynomial $G$. E.g. the massless one-loop bubble graph is a descendant of the one-loop bubble graph with one mass, which itself is a descendant of the full massive one-loop bubble. Those ancestors and descendants do not necessarily correspond to Feynman integrals in the original sense, since one can also consider polynomials $G$ which are not connected to graph polynomials anymore.
    
    As a trivial subtriangulation of an arbitrary triangulation one can choose one of its simplices. In doing so, one can relate the prefactors $C_\sigma (\nuu)$ to the problem where only one simplex is involved. For such problems, one can solve the Feynman integral easily as seen in corollary \ref{cor:Fox}. Therefore, for an unimodular triangulation one can find that the prefactors are simply given by
    \begin{align}
      C_{\sigma}(\nuu) = \Gamma(\Aas^{-1}\nuu)
    \end{align}
    which results in the following theorem:
       
    \begin{theorem}[Series representation of Feynman integrals] \label{thm:FeynSeries}
      Let $T$ be a regular, unimodular triangulation of the Newton polytope $\Delta_G=\Conv (\mathrm A)$ corresponding to a generalized Feynman integral $J_\Aa$. Then the generalized Feynman integral can be written as
      \begin{align}
        J_\Aa(\nuu,z) = \sum_{\sigma\in T} z_\sigma^{-\Aas^{-1}\nuu} \sum_{\lambda\in\mathbb N^{|\bar\sigma|}_0} \frac{\Gamma(\Aas^{-1}\nuu+\Aas^{-1}\Aabs\lambda)}{\lambda!} \frac{z_{\bar\sigma}^\lambda}{(-z_\sigma)^{\Aas^{-1}\Aabs\lambda}} \label{eq:seriesreprFI}
      \end{align}
      where the series have a common region of convergence. This representation holds for generic $\nuu\in\mathbb C^{n+1}$, which means that $\nuu$ has to be chosen such that the Feynman integral has no poles and none of the power series in (\ref{eq:seriesreprFI}) will be identical to zero.
    \end{theorem}
    
    \begin{example}
      To illustrate the series representation, we continue the example \ref{ex:1loopbubbleA} corresponding to figure \ref{fig:bubble1}. The point configuration for this Feynman graph was given by
      \begin{align}
        \Aa = \begin{pmatrix} 
              1 & 1 & 1 & 1 \\
              1 & 0 & 1 & 2 \\
              0 & 1 & 1 & 0
             \end{pmatrix} \qquad z = (1,1,m_1^2-p^2,m_1^2) \point
      \end{align}
      For the triangulation $T_1 = \{\{1,2,4\},\{2,3,4\}\}$ one obtains the series representation
      \begin{align}
        J_\Aa(\nuu,z) &= z_1^{-2 \nu_0 + \nu_1 + 2 \nu_2} z_2^{-\nu_2} z_4^{\nu_0 - \nu_1 - \nu_2} \sum_{\lambda\in\mathbb N_0} \frac{1}{\lambda !} \left(-\frac{z_1 z_3}{z_2 z_4}\right)^\lambda \Gamma (\nu_2+\lambda ) \nonumber\\
        &\qquad \Gamma (2 \nu_0-\nu_1-2 \nu_2-\lambda ) \Gamma (-\nu_0+\nu_1+\nu_2+\lambda ) \nonumber  \\
        &\qquad + z_4^{\nu_2-\nu_0} z_2^{-2 \nu_0+\nu_1+\nu_2} z_3^{2 \nu_0-\nu_1-2 \nu_2} \sum_{\lambda\in\mathbb N_0} \frac{1}{\lambda !} \left(-\frac{z_1 z_3}{z_2 z_4}\right)^\lambda  \Gamma (\nu_0-\nu_2+\lambda ) \nonumber \\
        &\qquad \Gamma (-2 \nu_0+\nu_1+2 \nu_2-\lambda ) \Gamma (2 \nu_0-\nu_1-\nu_2+\lambda ) \point
      \end{align}
      In the physical relevant limit $z\rightarrow (1,1,m_1^2-p^2,m_1^2)$ and $\nuu = (2-\epsilon,1,1)$ one can easily evaluate the series
      \begin{align}
        &J_\Aa(2-\epsilon,1,1,1,1,m_1^2-p^2,m_1^2) = (m_1^2)^{-\epsilon} \Gamma(1-2\epsilon)\Gamma(\epsilon) \HypF{1,\epsilon}{2\epsilon}{\frac{ m_1^2-p^2}{m_1^2}} \nonumber \\
        &\qquad +  (m_1^2-p^2)^{1-2\epsilon} (p^2)^{-1+\epsilon} \Gamma(1-\epsilon) \Gamma(2-2\epsilon)\Gamma(-1+2\epsilon) 
      \end{align}
      which agrees with the expected result. The series representation which can be obtained by the triangulation $T_2$, as well as the former result in example \ref{ex:1loopbubbleA}, are equivalent to this result by transformation rules of the ${}_2F_1$ function.      
    \end{example}
    
    Thus, we found a series representation for generalized Feynman integrals which admit an unimodular triangulation. Based on experience, it is reasonable to conjecture that every off-shell Feynman graph admits an unimodular triangulation. But also for non-unimodular triangulations, the Feynman integral can be referred back to a case which admits an unimodular triangulation. We will examine the case of non-unimodular triangulations and this relation in section \ref{ssec:nonunimodular}. Typically, a Feynman graph admits many different possibilities to triangulate its corresponding Newton polytope. Therefore, one usually obtains a large number of series representations. This is not surprising, since hypergeometric functions satisfy many transformation formulas and can be converted to other hypergeometric functions.
    
    Therefore, in practical computations one can choose a series representation, which converges fast for the given kinematics and evaluate the Feynman integral numerically by considering the first summands of every series.
    
    These series representations - like the whole GKZ approach -  are only valid for generalized Feynman integrals. In principle this series representation is also true for fixed values of $z$. But for non-generic values of $z$ it can happen, that the series in (\ref{eq:seriesreprFI}) do not converge anymore. We address this point in the following.

    \subsection{The Limit from Generalized Feynman Integrals to Non-Generic Feynman Integrals} \label{ssec:nongeneric}
    
    Up to this point, the theorems were statements about the generalized Feynman integral. The next natural question is how one can perform the limit to the non-generic, ``ordinary'' Feynman integral, i.e. the limit, where the extra variables, which are connected to the first Symanzik polynomial are set equal to $1$. In this limit the convergence behavior of the $\Gamma$-series can be changed. Consider a region $D\subseteq \mathbb C^{n+1}$ where the Feynman integral (\ref{eq:leepom}) converges for $\nuu\in D$. Since in this case the linear combination (\ref{eq:lincomb}) is still finite, there can arise only two problems: a) every series converges separately, but they do not have a common convergence region anymore or b) some of the $\Gamma$-series diverge, but the linear combination is still finite. In the first case a) the convergence criteria for the variables of the $\Gamma$-series $x_j = (z_{\bar\sigma})_j \prod_i (z_\sigma)_i^{-(\Aas^{-1} \Aabs)_{ij}}$ exclude each other for different $\sigma\in T$. In the second case b) the variables $x_j$ become constants after the limit (usually equals $1$), which can be outside of the convergence region. 
    
    Because of these possible issues, it can be difficult to perform the limit. Fortunately there is a strategy to tackle these problems in many cases by transformation formulas of hypergeometric functions. E.g. for the ${}_2F_1$ hypergeometric function, there is a well known transformation formula \cite{Olver.NationalINISTHandbookMathematical2010a}
    \begin{align}
      \HypF{a,b}{c}{z} &= \frac{\Gamma(c)\Gamma(c-a-b)}{\Gamma(c-a)\Gamma(c-b)} \HypF{a,b}{a+b-c+1}{1-z} \nonumber \\
      & + (1-z)^{c-a-b} \frac{\Gamma(c)\Gamma(a+b-c)}{\Gamma(a)\Gamma(b)} \HypF{c-a,c-b}{c-a-b+1}{1-z} \comma \label{eq:2F1trafo}
    \end{align}
    which can be applied to change a limit $x_j\rightarrow 1$ to the much simpler case of a limit $x_j\rightarrow 0$. We illustrate this method with an example. For the $2$-loop sunset graph with two different masses, inter alia there appears the hypergeometric series
    \begin{align}
      \phi_2 =& \sum_{k\in\mathbb N_0^4} (1-\epsilon)_{k_3+k_4} (\epsilon)_{k_1+2 k_2+k_3} (\epsilon -1)_{-k_1-k_2+k_4} (2-2 \epsilon)_{k_1-k_3-k_4} \nonumber \\
      &\qquad \frac{1}{k_1! k_2! k_3! k_4!}  \left(-\frac{z_1 z_6}{z_5 z_2}\right)^{k_1} \left(-\frac{z_4 z_6}{z_5^2}\right)^{k_2} \left(-\frac{z_2 z_7}{z_3 z_5}\right)^{k_3} \left(-\frac{z_2 z_8}{z_3 z_6}\right)^{k_4} 
    \end{align}
    where one has to consider the limit $(z_1,z_2,z_3,z_4,z_5,z_6,z_7,z_8)\rightarrow (1,1,1,m_2^2,m_1^2+m_2^2-p_1^2,m_1^2,m_2^2,m_1^2)$. In this limit it appears the term $\left(-1\right)^{k_4}$, which is not in the convergence region for small values of $\epsilon>0$ anymore. Therefore, we evaluate the $k_4$ series carefully and write
    \begin{align}
     \phi_2 &= \lim_{t\rightarrow 1} \sum_{(k_1,k_2,k_3)\in\mathbb N_0^3} (1-\epsilon)_{k_3}  (\epsilon-1)_{-k_1-k_2} (2-2 \epsilon)_{k_1-k_3} (\epsilon)_{k_1+2 k_2+k_3} \frac{1}{k_1! k_2! k_3!} \nonumber \\
      & \qquad\left(-x_1\right)^{k_1} \left(-x_1 x_2\right)^{k_2} \left(-x_2\right)^{k_3}  \HypF{-\epsilon+k_3+1,\epsilon-k_1-k_2-1}{2 \epsilon-k_1+k_3-1}{t}
    \end{align}
    where $x_i=\frac{ m_i^2}{m_1^2+m_2^2-p_1^2}$. With the transformation formula (\ref{eq:2F1trafo}) for the ${}_2F_1$ function, one can split the series in a convergent and a divergent part
    \begin{align}
      &\phi_2 = \sum_{(k_1,k_2,k_3)\in\mathbb N_0^3}  \frac{\Gamma (k_2+2 \epsilon -1) \Gamma (-k_1+k_3+2 \epsilon -1) }{ \Gamma (-k_1+3 \epsilon -2) \Gamma (k_2+k_3+\epsilon )} (1-\epsilon )_{k_3} (\epsilon -1)_{-k_1-k_2} (2-2 \epsilon )_{k_1-k_3} \nonumber \\
      & \quad (\epsilon )_{k_1+2 k_2+k_3} \frac{1}{k_1! k_2! k_3!} (-x_1)^{k_1} (-x_2)^{k_3} (-x_1 x_2)^{k_2} + \lim_{t\rightarrow 1} \sum_{(k_1,k_2,k_3,k_4)\in\mathbb N_0^4} \frac{(1-t)^{k_2+k_4+2 \epsilon -1} }{k_1! k_2! k_3! k_4! } \nonumber \\
      & \quad \frac{\Gamma (-k_2-2 \epsilon +1)  \Gamma (-k_1+k_3+2 \epsilon -1) \Gamma (k_2+k_3+\epsilon +k_4)\Gamma(k_2+2 \epsilon ) \Gamma (-k_1+3 \epsilon -2 +k_4) }{ \Gamma (k_3-\epsilon +1) \Gamma (-k_1-k_2+\epsilon -1) \Gamma(k_2+k_3+\epsilon ) \Gamma(k_2+2 \epsilon +k_4)\Gamma (-k_1+3 \epsilon -2)} \nonumber \\
      & \quad (\epsilon )_{k_1+2 k_2+k_3}  (1-\epsilon )_{k_3} (\epsilon -1)_{-k_1-k_2} (2-2 \epsilon )_{k_1-k_3} (-x_1)^{k_1} (-x_2)^{k_3} (-x_1 x_2)^{k_2} \nonumber  \allowdisplaybreaks \\
      &=\sum_{(k_1,k_2,k_3)\in\mathbb N_0^3}  \frac{\Gamma (k_2+2 \epsilon -1) \Gamma (-k_1+k_3+2 \epsilon -1) }{ \Gamma (-k_1+3 \epsilon -2) \Gamma (k_2+k_3+\epsilon )} (1-\epsilon )_{k_3} (\epsilon -1)_{-k_1-k_2} (2-2 \epsilon )_{k_1-k_3} \nonumber \\
      & \quad (\epsilon )_{k_1+2 k_2+k_3}  \frac{1}{k_1! k_2! k_3!} (-x_1)^{k_1} (-x_2)^{k_3} (-x_1 x_2)^{k_2} + \lim_{t\rightarrow 1} (1-t)^{2 \epsilon -1} \sum_{(k_1,k_3)\in\mathbb N_0^2} \frac{  (-x_1)^{k_1} (-x_2)^{k_3}}{k_1! k_3!} \nonumber \\
      & \quad \frac{\Gamma (-2 \epsilon +1)  \Gamma (-k_1+k_3+2 \epsilon -1)}{ \Gamma (1-\epsilon) \Gamma (\epsilon -1)} (\epsilon )_{k_1+k_3}   (2-2 \epsilon )_{k_1-k_3} \point
    \end{align}
    
    Comparing the divergent part with the other $\Gamma$-series, which occurs in the calculation of the sunset graph with two masses, one can find another divergent series which exactly cancels this divergence. This cancellation always has to happen, since the linear combination has to be finite.
    
    Therefore, one can derive a convergent series representation also for non-generic Feynman integrals by considering convenient transformation formulas of hypergeometric functions. In principle, this procedure will work in general, but it can be necessary to involve more complicated transformation formulas than the well-studied ${}_2F_1$ transformation and also the cancellation of the divergences may be not so obvious as in this case. Fortunately, many graphs (also for $L>1$) can even be evaluated with the transformation formula of the ${}_2F_1$ function.
    
    In fact, this limit can reduce the dimension of the solution space, which is the expected behavior. For generic variables $z\in\mathbb C^N$ the dimension of the solution space is equal to $\vol_0 \Delta_G$ according to theorem \ref{thm:holrank}. In contrast for non-generic values of $z\in\mathbb C^N$ the dimension of the solution space is equal to the Euler characteristic $(-1)^n \chi ((\mathbb C^\star)^n \setminus \{G=0\})$ \cite{Bitoun.Bogner.Klausen.PanzerFeynmanIntegralRelations2019}, which is in general smaller or equal to the volume of the Newton polytope 
     \begin{align}
      (-1)^n \chi((\mathbb C^\star)^n \setminus \{G=0\}) \leq \vol_0 \Delta_G \point
    \end{align}
    Thus, by calculating the Euler characteristic we can count the expected dependencies in the limit from generalized to non-generic Feynman integrals.    According to \cite{Bitoun.Bogner.Klausen.PanzerFeynmanIntegralRelations2019} it is meaningful to define the number of master integrals as the dimension of the solution space. Therefore, one can see the linear combination of $\Gamma$-series (\ref{eq:lincomb}) as a procedure similar to a decomposition of a general Feynman integral into a basis of master integrals. This analogy mirrors also the existence of different decompositions (which corresponds to different triangulations) and their transformation into each other by shift relations \cite{Kniehl.TarasovFindingNewRelationships2012}.

    \subsection{The $\epsilon$ Expansion of Hypergeometric Series Representations} \label{ssec:epsilon}
    
    In the calculus of dimensional regularization \cite{tHooft.VeltmanRegularizationRenormalizationGauge1972} one is usually interested in the Laurent expansion of the Feynman integrals around $\epsilon=0$ where $d=4-2\epsilon$. Due to the theorem \ref{thm:meromorphic} one can relate this task to the Taylor expansion of the hypergeometric series representation. Thus, one has simply to differentiate the Horn hypergeometric series. As pointed out in \cite{Bytev.Kniehl.MochDerivativesHorntypeHypergeometric2017}, the derivatives with respect to parameters of Horn hypergeometric series are again Horn hypergeometric series of higher degree. By the identities $(a)_{m+n} = (a)_m (a+m)_n$ and $(a)_{rn} = r^{rn} \prod_{j=0}^{r-1} \left(\frac{a+j}{r}\right)_n$ for $r\in\mathbb Z_{>0}$ one can reduce all derivatives to two cases \cite{Bytev.Kniehl.MochDerivativesHorntypeHypergeometric2017}
    \begin{align}
      \frac{\partial}{\partial a} \sum_{n=0}^\infty B(n) (a)_n x^n &= x \sum_{k=0}^\infty\sum_{n=0}^\infty B(n+k+1) \frac{(a+1)_{n+k} (a)_k}{(a+1)_k} x^{n+k} \\
      \frac{\partial}{\partial a} \sum_{n=0}^\infty B(n) (a)_{-n} x^n &= -x \sum_{k=0}^\infty\sum_{n=0}^\infty B(n+k+1) \frac{(a)_{-n-k-1} (a)_{-k-1}}{(a)_{-k}} x^{n+k} \point
    \end{align}
    Thus, Horn hypergeometric functions do not only appear as solutions of Feynman integrals with unimodular triangulations, but also in every coefficient of the Laurent expansion of those Feynman integrals. Therefore, the class of Horn hypergeometric functions is sufficient to describe almost all Feynman integrals and their Laurent expansion. Hence, it is not surprising that also the combinatorial structure of Feynman integrals is reflected in the Horn hypergeometric functions. For instance, the relations between different Feynman integrals can be derived by transformation formulas of hypergeometric functions and vice versa \cite{Kniehl.TarasovFindingNewRelationships2012}.
    
    Another way to expand the Horn hypergeometric functions around $\epsilon = 0$ in some cases could be the approach of $S$- and $Z$-sums \cite{Moch.Uwer.WeinzierlNestedSumsExpansion2002}, which are related to multiple polylogarithms and related functions. Unfortunately, many examples, which can be generated by the GKZ approach, belong not to the known algorithms given in \cite{Moch.Uwer.WeinzierlNestedSumsExpansion2002}.

    \subsection{Applications to other Parametric Representations} \label{ssec:otherparametric}
    
    The GKZ mechanism with the resulting series representation used above is not limited to the Feynman integral in the Lee-Pomeransky representation. It is a method which can be used for all integrals of Euler-Mellin type including one polynomial\footnote{Following the approach described in \cite{Berkesch.Forsgard.PassareEulerMellinIntegrals2011} a generalization to Euler-Mellin integrals with several polynomials is also possible.}. Thus, there are more applications in the Feynman integral calculus. For example the Feynman parametric representation (\ref{eq:FI1}) is of that form, if either the first Symanzik polynomial $U$ or the second Symanzik polynomial $F$ drops out. 
    
    This is the case e.g. for so called ``marginal'' Feynman integrals \cite{Bourjaily.McLeod.vonHippel.WilhelmBoundedBestiaryFeynman2019}, where $\omega=d/2$. In this case the first Symanzik polynomial drops out of the representation (\ref{eq:FI1}) and one obtains the integral $\int_{\mathbb R^{n-1}_+}\d[\tilde x] \tilde x^{\tilde\nu-1} \tilde F^{-d/2}$ with  $\tilde x = (x_1,\ldots,x_{n-1})$, $\tilde \nu = (\nu_1,\ldots,\nu_{n-1})$ and $\tilde F = F|_{x_n\rightarrow 1}$. For instance all ``banana''-graphs are marginal for $\nu_i=1$ and $d=2$. 
    
    In contrast the second Symanzik polynomial drops out for $\omega=0$, which is the well-studied case of periods (e.g. \cite{BrownPeriodsFeynmanIntegrals2009}) $\int_{\mathbb R^{n-1}_+}\d[\tilde x] \tilde x^{\tilde\nu-1} \tilde U^{-d/2}$ with  $\tilde U = U|_{x_n\rightarrow 1}$. Note that this integral is highly non-generic from the perspective of the GKZ approach. Thus, the case where the second Symanzik polynomial $F$ remains is usually much easier, since one has not to introduce extra variables in the first Symanzik polynomial $U$.
    
    Last but not least, also the Baikov representation \cite{GrozinIntegrationPartsIntroduction2011, Bitoun.Bogner.Klausen.PanzerFeynmanIntegralRelations2019} is a possible candidate to apply the GKZ approach as well.

    \subsection{Non-Unimodular Triangulations of Feynman Polytopes} \label{ssec:nonunimodular}
    
    The treatments above were specialized to unimodular triangulations only. This strategy has various reasons. Without much effort, one could extend theorem \ref{thm:solutionspace} also to the case of non-unimodular triangulations and write the Feynman integral as a linear combination of $\Gamma$-series as in equation (\ref{eq:lincomb}). In contrast, for a non-unimodular triangulation one can not determine the meromorphic functions $C_\sigma(\nuu)$ such as easy as in theorem \ref{thm:FeynSeries}. Nevertheless, one can reduce Feynman integrals without unimodular triangulations to subtriangulations as described in section \ref{ssec:series}.
    
    However, after checking common Feynman integrals up to three loops, we are leaded to the conjecture that all off-shell Feynman graphs admit at least one unimodular triangulation. This conjecture seems also likely by considering the very specific form of Newton polytopes $\Delta_G$ arising in Feynman integrals (lemma \ref{lem:structurepolytopes}). 
      
    In contrast some on-shell graphs do not admit an unimodular triangulation. For instance the on-shell full massive 1-loop bubble with $G=x_1+x_2+m_1^2 x_1^2+m_2^2 x_2^2$ does not allow an unimodular triangulation. However, since the off-shell fully massive 1-loop bubble admits unimodular triangulations, one can treat the on-shell Feynman integral as a limit of the off-shell version.
    
    This behavior holds in general: one can always add monomials to the Lee-Pomeransky polynomial $G\rightarrow G^\prime$, such that the Newton polytope $\Delta_{G^\prime}$ admits an unimodular triangulation. This can be seen by a result of Knudsen et al. \cite{Kempf.Knudsen.MumfordConstructionNicePolyhedral1973}:
    
    \begin{lemma}[\cite{Kempf.Knudsen.MumfordConstructionNicePolyhedral1973, BrunsPolytopesRingsKTheory2009}]
      For every lattice polytope $P$ there is an integer $k\in\mathbb N$ such that the dilated polytope $kP:= \{k \mu |\mu\in P\}$ admits an unimodular triangulation.
    \end{lemma}
    
    Thus, if the Newton polytope $\Delta_G$ of a Feynman integral does not admit an unimodular triangulation, one can dilate the polytope $\Delta_G$ until the polytope $k\Delta_G$ admits an unimodular triangulation. Then the original Feynman integral can be obtained as a limit where all additional vertices vanish and where we scale the propagator powers by $k$ and the whole integral by $k^n$
    \begin{align}
      J_\Aa(\nuu,z) = k^n \Gamma(\nu_0) \int_{\mathbb R_+^n} \d[x] x^{k \nu-1} G_z(x^k)^{-\nu_0} \point
    \end{align}
    Thus, one can always write the Feynman integral as a limit of another integral where its Newton polytope allows an unimodular triangulation. That is the reason, why we can focus only on the unimodular triangulations and vindicates the procedure above.\footnote{After completion of the present article, it has been found that a similar result as in theorem \ref{thm:FeynSeries} can also be achieved for regular, non-unimodular triangulation. This will be discussed in more detail in a future article.}

  \section{Advanced Example: Full Massive Sunset} \label{sec:ex}
  
  To illustrate the GKZ method stated above, as well as to show the power of this approach, we calculate a series representation of the sunset Feynman integral with three different masses according to figure \ref{fig:sunset}. The corresponding Feynman graph consists in $n=3$ edges and the Lee-Pomeransky polynomial includes $N=10$ monomials
  \begin{align}
    G &= x_1 x_2 + x_1 x_3 + x_2 x_3 + (m_1^2+m_2^2+m_3^2-p^2) x_1 x_2 x_3 \nonumber \\
    &\qquad + m_1^2 x_1^2 (x_2+x_3) + m_2^2 x_2^2 (x_1+x_3) + m_3^2 x_3^2 (x_1+x_2) \point
  \end{align}
  
  \begin{figure}[ht]
    \begin{center}
      \vspace{-1cm}
      \includegraphics[width=.38\textwidth]{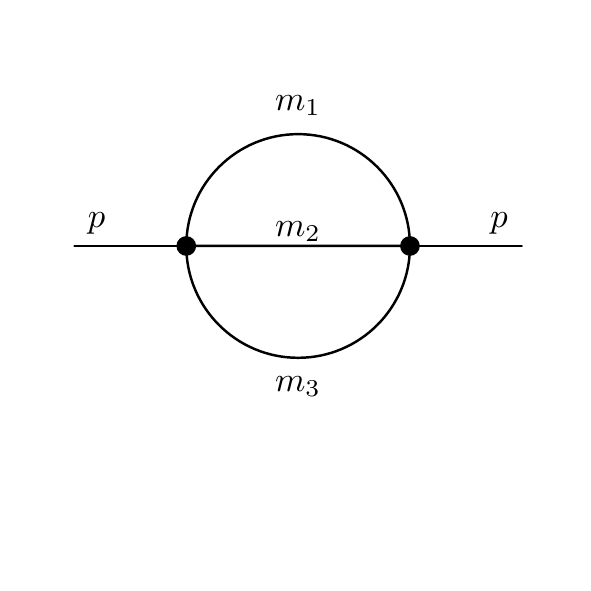}
      \vspace{-2cm}
    \end{center}
    \caption{The $2$-loop $2$-point function (sunset graph) with three different masses.}
    \label{fig:sunset}
  \end{figure}
  
  In the representation of equation (\ref{eq:Gzx}) we encode this polynomial by
  \begin{align}
    \Aa &= \begin{pmatrix}
	      1 & 1 & 1 & 1 & 1 & 1 & 1 & 1 & 1 & 1 \\
	      0 & 1 & 1 & 0 & 1 & 0 & 1 & 2 & 1 & 2 \\
	      1 & 0 & 1 & 1 & 0 & 2 & 1 & 0 & 2 & 1 \\
	      1 & 1 & 0 & 2 & 2 & 1 & 1 & 1 & 0 & 0
	    \end{pmatrix} \\
    z &= (1,1,1,m_3^2,m_3^2,m_2^2,m_1^2+m_2^2+m_3^2-p_1^2,m_1^2,m_2^2,m_1^2) \point
  \end{align}

  The rank of the kernel of $\Aa$ is equal to $r=N-n-1=6$ and therefore we will expect $6$-dimensional $\Gamma$-series. Moreover, the polytope $\Delta_G=\Conv (\mathrm A)$ has the volume $\vol_0(\Conv( \mathrm A))=10$ (calculated with polymake \cite{polymake:2000}), which leads to $10$ basis solutions, and there are $826$ different ways for a regular triangulation of the Newton polytope $\Delta_G$, where $466$ of those triangulations are unimodular. We choose the unimodular triangulation (calculated with TOPCOM \cite{Rambau.topcom})
  \begin{align}
    T_{152} &= \{\{3,6,7,9\},\{3,7,9,10\},\{3,7,8,10\},\{2,5,7,8\},\{2,3,7,8\},\nonumber \\
    &\qquad \{2,4,5,7\},\{1,4,6,7\},\{1,2,4,7\},\{1,3,6,7\},\{1,2,3,7\}\}
  \end{align}
  in order to get series, which converge fast for highly relativistic kinematics $m_i^2 \ll m_1^2+m_2^2+m_3^2-p^2$. Further, we set $\nu_i=1$ and $d=4-2\epsilon$. 
    
  In the limit $z\rightarrow(1,1,1,m_3^2,m_3^2,m_2^2,m_1^2+m_2^2+m_3^2-p_1^2,m_1^2,m_2^2,m_1^2)$ the series $\phi_1$, $\phi_3$, $\phi_5$, $\phi_6$, $\phi_8$ and $\phi_9$ are divergent for small values of $\epsilon>0$. By the method described in section \ref{ssec:nongeneric} one can split all these series by the transformation formula for the ${}_2F_1$ hypergeometric function in a convergent and a divergent part. The divergent parts of these series cancel each other. In doing so the resulting $\Gamma$-series have linear dependences and the dimension of the solution space will reduce from $10$ to $7$.
  
  By applying all these steps one arrives at the following series representation of the full massive sunset integral
  \begin{align}
    I_\Aa (\nuu,z) &= \frac{s^{1-2\epsilon}}{\Gamma(3-3\epsilon)} \left[ x_2^{1-\epsilon} \phi_1 + (x_1 x_2)^{1-\epsilon} \phi_2 + x_1^{1-\epsilon} \phi_3 + (x_1 x_3)^{1-\epsilon} \phi_4 + x_1^{1-\epsilon} \phi_5 \right. \nonumber \\
    &\quad \left. + x_3^{1-\epsilon} \phi_6 + (x_2 x_3)^{1-\epsilon} \phi_7 + x_3^{1-\epsilon} \phi_8 + x_2^{1-\epsilon} \phi_9 + \phi_{10} \right]
  \end{align}
  where the $\Gamma$-series are given by
  \begin{align}
    \phi_1 &= \sum_{k_2,k_3,k_4,k_5,k_6=0}^\infty\frac{ (-x_2)^{k_2} (-x_3)^{k_3} (-x_2 x_3)^{k_4} (-x_1 x_2)^{k_5} (-x_1)^{k_6} }{k_2! k_3! k_4! k_5! k_6!} \nonumber\\
    &\quad \Gamma (k_2-3 \epsilon +3) \Gamma (k_2+k_3+k_4-k_6-2 \epsilon +3) \Gamma (k_3-k_5-k_6-\epsilon +1) \nonumber \\
    &\quad \frac{\Gamma (k_2+k_3+2 k_4+2 k_5+k_6+\epsilon ) \Gamma (k_4+k_5+2 \epsilon -1) \Gamma (-k_2-k_3-k_4+k_6+2 \epsilon -2)}{\Gamma (k_2+k_4+k_5-\epsilon +2) \Gamma (k_3+k_4-k_6+\epsilon )}  \allowdisplaybreaks \nonumber\\
    \phi_2 &= \sum_{k_1,k_2,k_3,k_4,k_5,k_6=0}^\infty\frac{ (-x_1)^{k_1+k_5} (-x_2)^{k_2+k_6}  (-x_1 x_3)^{k_3} (-x_2 x_3)^{k_4}}{k_1! k_2! k_3! k_4! k_5! k_6!} \nonumber\\
    &\quad \Gamma (k_1+k_2+2 k_3+2 k_4+k_5+k_6+1) \Gamma (k_1+k_2-3 \epsilon +3) \nonumber \\
    &\quad \Gamma (-k_2-k_4+k_5-k_6+\epsilon -1) \Gamma (-k_1-k_3-k_5+k_6+\epsilon -1) \allowdisplaybreaks \nonumber\\
    \phi_3 &= \sum_{k_1,k_3,k_4,k_5,k_6=0}^\infty\frac{(-x_1)^{k_1}  (-x_1 x_3)^{k_3} (-x_3)^{k_4} (-x_1 x_2)^{k_5} (-x_2)^{k_6} }{k_1! k_3! k_4! k_5! k_6!} \nonumber\\
    &\quad \Gamma (k_1-3 \epsilon +3) \Gamma (k_1+k_3+k_4-k_6-2 \epsilon +3) \Gamma (k_4-k_5-k_6-\epsilon +1) \nonumber \\
    &\quad \frac{\Gamma (k_1+2 k_3+k_4+2 k_5+k_6+\epsilon ) \Gamma (k_3+k_5+2 \epsilon -1) \Gamma (-k_1-k_3-k_4+k_6+2 \epsilon -2)}{\Gamma (k_1+k_3+k_5-\epsilon +2) \Gamma (k_3+k_4-k_6+\epsilon )} \allowdisplaybreaks \nonumber\\
    \phi_4 &= \sum_{k_1,k_2,k_3,k_4,k_5,k_6=0}^\infty\frac{ (-x_1)^{k_1+k_3}  (-x_3)^{k_2+k_6} (-x_1 x_2)^{k_4}  (-x_2 x_3)^{k_5}}{k_1! k_2! k_3! k_4! k_5! k_6!} \nonumber\\
    &\quad \Gamma (k_1+k_2+k_3+2 k_4+2 k_5+k_6+1) \Gamma (k_1+k_2-3 \epsilon +3) \nonumber \\
    &\quad \Gamma (-k_2+k_3-k_5-k_6+\epsilon -1) \Gamma (-k_1-k_3-k_4+k_6+\epsilon -1) \allowdisplaybreaks \nonumber\\
    \phi_5 &= \sum_{k_1,k_2,k_3,k_4,k_5=0}^\infty\frac{ (-x_1)^{k_1} (-x_1 x_3)^{k_2} (-x_3)^{k_3}  (-x_1 x_2)^{k_4} (-x_2)^{k_5} }{k_1! k_2! k_3! k_4! k_5!} \nonumber\\
    &\quad \Gamma (k_1+k_2+k_3-k_5-2 \epsilon +2) \Gamma (-k_2-k_3+k_5-\epsilon +1) \Gamma (-k_1-k_2-k_4+\epsilon -1)  \nonumber \\
    &\quad \frac{\Gamma (k_1+2 k_2+k_3+2 k_4+k_5+\epsilon ) \Gamma (k_2+k_4+2 \epsilon -1) \Gamma (-k_1-k_2-k_3+k_5+2 \epsilon -1)}{\Gamma (-k_3+k_4+k_5+\epsilon ) \Gamma (-k_1+3 \epsilon -2)} \allowdisplaybreaks \nonumber\\
    \phi_6 &= \sum_{k_2,k_3,k_4,k_5,k_6=0}^\infty\frac{ (-x_3)^{k_2} (-x_2)^{k_3}  (-x_1)^{k_4}   (-x_2 x_3)^{k_5} (-x_1 x_3)^{k_6} }{k_2! k_3! k_4! k_5! k_6!} \nonumber\\
    &\quad \Gamma (k_2-3 \epsilon +3) \Gamma (k_2+k_3-k_4+k_5-2 \epsilon +3) \Gamma (k_3-k_4-k_6-\epsilon +1) \nonumber \\
    &\quad \frac{\Gamma (k_2+k_3+k_4+2 k_5+2 k_6+\epsilon ) \Gamma (-k_2-k_3+k_4-k_5+2 \epsilon -2) \Gamma (k_5+k_6+2 \epsilon -1)}{\Gamma (k_2+k_5+k_6-\epsilon +2) \Gamma (k_3-k_4+k_5+\epsilon )} \allowdisplaybreaks \nonumber\\
    \phi_7 &= \sum_{k_1,k_2,k_3,k_4,k_5,k_6=0}^\infty\frac{ (-x_2)^{k_1+k_3}(-x_3)^{k_2+k_5} (-x_1 x_2)^{k_4}  (-x_1 x_3)^{k_6}}{k_1! k_2! k_3! k_4! k_5! k_6!} \nonumber\\
    &\quad \Gamma (k_1+k_2+k_3+2 k_4+k_5+2 k_6+1) \Gamma (k_1+k_2-3 \epsilon +3) \nonumber \\
    &\quad  \Gamma (-k_1-k_3-k_4+k_5+\epsilon -1) \Gamma (-k_2+k_3-k_5-k_6+\epsilon -1) \allowdisplaybreaks \nonumber\\
    \phi_8 &= \sum_{k_1,k_3,k_4,k_5,k_6=0}^\infty\frac{(-x_3)^{k_1}  (-x_2)^{k_3} (-x_1)^{k_4} (-x_2 x_3)^{k_5}  (-x_1 x_3)^{k_6} }{k_1! k_3! k_4! k_5! k_6!} \nonumber\\
    &\quad \Gamma (k_1+k_3-k_4+k_5-2 \epsilon +2) \Gamma (-k_3+k_4-k_5-\epsilon +1) \Gamma (-k_1-k_5-k_6+\epsilon -1) \nonumber \\
    &\quad \frac{\Gamma (k_1+k_3+k_4+2 k_5+2 k_6+\epsilon ) \Gamma (-k_1-k_3+k_4-k_5+2 \epsilon -1) \Gamma (k_5+k_6+2 \epsilon -1)}{\Gamma (-k_3+k_4+k_6+\epsilon ) \Gamma (-k_1+3 \epsilon -2)} \allowdisplaybreaks \nonumber\\
    \phi_9 &= \sum_{k_1,k_2,k_3,k_4,k_6=0}^\infty\frac{  (-x_2)^{k_1} (-x_3)^{k_2} (-x_2 x_3)^{k_3} (-x_1 x_2)^{k_4}  (-x_1)^{k_6} }{k_1! k_2! k_3! k_4! k_6!} \nonumber\\
    &\quad \Gamma (k_1+k_2+k_3-k_6-2 \epsilon +2) \Gamma (-k_2-k_3+k_6-\epsilon +1) \Gamma (-k_1-k_3-k_4+\epsilon -1)  \nonumber \\
    &\quad \frac{\Gamma (k_1+k_2+2 k_3+2 k_4+k_6+\epsilon ) \Gamma (k_3+k_4+2 \epsilon -1) \Gamma (-k_1-k_2-k_3+k_6+2 \epsilon -1)}{\Gamma (-k_2+k_4+k_6+\epsilon ) \Gamma (-k_1+3 \epsilon -2)} \allowdisplaybreaks \nonumber\\
    \phi_{10} &= \sum_{k_1,k_2,k_3,k_4,k_5,k_6=0}^\infty\frac{ (-x_3)^{k_1+k_2} (-x_2)^{k_3+k_5} (-x_1)^{k_4+k_6} }{k_1! k_2! k_3! k_4! k_5! k_6!} \nonumber\\
    &\quad \Gamma (k_2-k_3+k_4-k_5-\epsilon +1) \Gamma (k_1+k_3-k_4-k_6-\epsilon +1) \nonumber \\
    &\quad  \Gamma (-k_1-k_2+k_5+k_6-\epsilon +1) \Gamma (k_1+k_2+k_3+k_4+k_5+k_6+2 \epsilon -1)
  \end{align}
  with $x_i = \frac{m_i^2}{m_1^2+m_2^2+m_3^2-p^2}$ and $s=m_1^2+m_2^2+m_3^2-p^2$. All these series converge for small values of $x_i$ and the series representation can be obtained by a very simple algorithm, which is a straightforward implementation of the steps described in section \ref{sec:FeynHyp}. In fact, some of these $\Gamma$-series are related to each other. One can reduce the whole system only to $\phi_1, \phi_2, \phi_5$ and $\phi_{10}$ by the relations $\phi_1(x_1,x_2,x_3) = \phi_3 (x_2,x_1,x_3) = \phi_6(x_1,x_3,x_2)$, $\phi_2(x_1,x_2,x_3) = \phi_4 (x_1,x_3,x_2) = \phi_7 (x_3,x_2,x_1)$ and $\phi_5 (x_1,x_2,x_3) = \phi_8 (x_2,x_3,x_1) = \phi_9 (x_2,x_1,x_3)$. By these relations one can also verify the expected symmetry of the Feynman integral under the permutation $x_1 \leftrightarrow x_2 \leftrightarrow x_3$.
   
  In order to expand the Feynman integral $I_\Aa$ for small values of $\epsilon>0$ one can use the methods described in section \ref{ssec:epsilon} or alternatively by expanding each $\Gamma$-function separately. The latter requires to distinguish between positive and negative integers in the argument of the $\Gamma$-function
  \begin{align}
    \Gamma(b\epsilon+n) & \stackrel{n\in\mathbb{Z}_{\geq 1}}{=} \Gamma(n) \left[1+b\epsilon \psi_0(n)+\frac{b^2\epsilon^2}{2} \left(\psi_0(n)^2+\psi_1(n) \right) \right.\nonumber \\ 
    & \left. + \frac{b^3\epsilon^3}{6} \left(\psi_0(n)^3 + 3\psi_0(n)\psi_1(n) + \psi_2(n)\right) +\order{\epsilon^4} \right] \nonumber \\
    \Gamma(b\epsilon+n) & \stackrel{n\in\mathbb Z_{\leq 0}}{=} \frac{(-1)^n}{\Gamma(-n+1)} \left[\frac{1}{b\epsilon} + \psi_0(1-n) + \frac{b\epsilon}{2} \left(2\zeta_2 + \psi_0(1-n)^2 -\psi_1(1-n)\right) \right. \nonumber \\
    & \left. + \frac{b^2\epsilon^2}{6} \left(\psi_0(1-n)^3+6 \zeta_2 \psi_0(1-n)-3\psi_0(1-n)\psi_1(1-n) + \psi_2(1-n) \right) +\order{\epsilon^3} \right] \textrm{.}
  \end{align}
  By the distinction of cases between positive and negative arguments in the $\Gamma$ functions many terms arise, which are easily manageable by a CAS but which are space-consuming in print, wherefore we omit 
  to state these results here. The correctness of these results was checked numerically by FIESTA \cite{SmirnovFIESTAOptimizedFeynman2016} with arbitrary kinematics and masses, satisfying $x_i< 0.5$. For small values of $x_i$ the resulting series converges fast, such that for a good approximation one only has to take the first summands into account. An upper bound for the errors in a finite summation can be estimated by a majorant geometric series, similar to the procedure in theorem \ref{thm:GammaConverge}. Vice versa one can determine the number of required summands for a given error bound. Furthermore, the summands of a Horn hypergeometric series always have a rational ratio. Thus, in a numerical calculation one only has to evaluate rational functions. In this way Horn hypergeometric series have a relatively simple and controllable numerical behavior in the case $|x_i|\ll 1$.

  \section{Conclusion and Outlook}
    
  We showed in this article that Feynman integrals can be described as hypergeometric functions. Namely we showed that a) every generalized Feynman integral is an $\Aa$-hypergeometric function, that b) every generalized Feynman integral which admits a regular, unimodular triangulation has a representation in Horn hypergeometric functions and that c) all scalar, dimensional regularizable and Euclidean Feynman integrals can be written at least as a limit of a linear combination of Horn hypergeometric functions. Furthermore, the latter (c) is also true for all coefficients in a Laurent expansion of the Feynman integral in a dimensional or analytic regularization.
    
  Since hypergeometric functions are mostly represented in terms of integrals including polynomials, Mellin-Barnes integrals including $\Gamma$-functions or series including Pochhammer symbols, it is not surprising that also Feynman integrals appear normally in one of those representations. From the perspective of general hypergeometric functions the common ground of these representations is an integer matrix $\Aa\in\mathbb Z^{(n+1)\times N}$ or equivalently a Newton polytope $\Delta_G = \Conv(\mathrm A)$. By the triangulation of this Newton polytope we derived an analytic formula for a hypergeometric series representation of the Feynman integral. As there are in general many different ways to triangulate a polytope, there are also many different series representations possible. Similar to the Feynman integral those hypergeometric series satisfy different relations between each other and can therefore also be transformed in equivalent representations. However, the series representations can differ in their convergence behavior. Thus, there can be series representations which converge fast for given kinematics so that we sometimes only need the first summands in order to reach a high precision numerical approximation.
    
  For the purpose of a practical usage of this concept, we discussed possible obstacles which can appear in the concrete evaluation and gave some strategies to solve them. The procedure described in section \ref{sec:FeynHyp} and illustrated in section \ref{sec:ex} works similarly in all cases and can be read as an algorithm.
    
  Besides numerical applications, there are structurally interesting implications for the Feynman integral. Since in both subjects similar questions appear, the hypergeometric perspective opens new ways to analyze the Feynman integral. For instance the singularities of hypergeometric functions, and therefore simultaneously the singularities of Feynman integrals, are given by the $\Aa$-discriminants and the coamoebas.
    
  The series representation given in the present work, is only one of many consequences that we can obtain by the connection between Feynman integrals and hypergeometric functions. On the one hand we outlined in section \ref{ssec:otherparametric} that we also can apply the GKZ approach to other integrals appearing in the Feynman calculus. On the other hand one can also solve the GKZ system with other functions to obtain representations of the Feynman integral e.g. in terms of other integrals. This leads us to the question, if one can relate Feynman integrals with the GKZ approach to other Feynman integrals in a consistent way, in order to get linear relations between Feynman integrals similar to the IBP approach.
    
  Furthermore, it would be desirable to get a better understanding of the relations between Feynman integrals and their diagrams sharing parts of the matrices $\Aa$. The simplest connection between those Feynman integrals are the ancestors and descendants of Feynman integrals mentioned in section \ref{ssec:series}. Lastly, we leave several technical questions open to solve in future work. For instance, a proof of the conjecture that off-shell Feynman integrals always admit unimodular triangulations together with a better understanding of the behavior in non-generic cases as described in section \ref{ssec:nongeneric}, would show that all off-shell Feynman integrals can be written as Horn hypergeometric functions and not only as a limit of Horn hypergeometric functions. 
    
  In this article we suggested a new perspective of an old idea to characterize Feynman integrals. This hypergeometric perspective can be useful for practical computation but also for a better conceptual understanding of Feynman integrals. Still, this is merely a starting point for future research on the correspondence between Feynman integrals and hypergeometric functions in the GKZ theory.
    
  \pagebreak
 
  \section{Appendix}
  
    \subsection{Convex Polytopes \& Triangulations}
  
    Convex polytopes constitute one of the main concepts in the present paper and are the key link between Feynman integrals and hypergeometric functions. They characterize the whole structure of a Feynman graph, which is necessary to evaluate the Feynman integral as a function of kinematics and masses. Furthermore they determine the region of convergence of the Feynman integral (theorem \ref{thm:FIconvergence}), the poles and the meromorphic continuation of the Feynman integral (theorem \ref{thm:meromorphic}) and their triangulations give series representations (theorem \ref{thm:FeynSeries}) in terms of Horn hypergeometric series.
    
    In section \ref{ssec:interlude} we already invented the most important concepts which are essential for this approach. At this point, we give some more details which are closely related to this approach and which give an illustrative viewpoint on polytopes. For further treatments we refer to \cite{BrondstedIntroductionConvexPolytopes1983, Henk.Richter-Ge.ZieglerBasicPropertiesConvex2004b, BrunsPolytopesRingsKTheory2009, DeLoera.Rambau.SantosTriangulationsStructuresAlgorithms2010}. \spacepar
    
    In the mathematical interlude in section \ref{ssec:interlude} we introduced the two ways to represent a polytope (by vertices and by the intersection of half-spaces). For general polytopes the transformation of these two representations can be complicated. However for simplices there is a simple connection.
    \begin{lemma} \label{lem:simplexfacets}
      Let $P_\bigtriangleup = \Conv (\mathrm A)$ be a full dimensional simplex with no internal points and let $\Aa = \begin{pmatrix} 1 \\ \mathrm A \end{pmatrix}$ be as before. Then the matrix $\Aa$ describes the relative interior of the polytope $\mu \in \relint (P_\bigtriangleup) \Leftrightarrow \Aa^{-1} \begin{pmatrix} 1 \\ \mu \end{pmatrix} > 0$. Furthermore the $i$-th row of $\Aa^{-1} \begin{pmatrix} 1 \\ \mu \end{pmatrix} = 0$ describes the facet, which is opposite to the point defined by the $i$-th column of $\Aa$.
    \end{lemma}
    \begin{proof}
      Since $P_\bigtriangleup$ is a full dimensional simplex it is $\det \Aa \neq 0$ and one can write
      \begin{align}
        P_\bigtriangleup = \left\{\mu\in\mathbb R^n \left| \begin{pmatrix} 1 \\ \mu \end{pmatrix} = \Aa k, k_i \geq 0 \right\}\right. =  \left\{ \mu\in\mathbb R^n \left| \Aa^{-1} \begin{pmatrix} 1 \\ \mu \end{pmatrix} \geq 0 \right\}\right. \point
      \end{align}
      Thus for the relative interior of a simplex it is $\Aa^{-1} \begin{pmatrix} 1 \\ \mu \end{pmatrix} > 0$ and every of the $n+1$ rows of $\Aa^{-1} \begin{pmatrix} 1 \\ \mu \end{pmatrix} = 0$ describes a facet of the simplex. 
   
      Vertices of the simplex are the intersection of $n$ facets. Thus, a vertex $v$ of $P_\bigtriangleup$ is the solution of $n$ rows of the linear equation system $\Aa^{-1} x = 0$. Clearly the $k$-th column of $\Aa$ solves the system $\Aa^{-1} x =0$ except of the $k$-th row. Hence, the intersection of $n$ facets is the $k$-th column of $\Aa$, where $k$ is the index which belongs to the facet which is not involved in the intersection. In a simplex that means that the $k$-th facet is opposite to the $k$-th vertex.
    \end{proof}
    
    In the case of a general convex polytope the connection between vertices and intersection of halfspaces is more complex. Nevertheless, one can establish an inequality which holds for all points in the polytope.
    
    \begin{lemma} \label{lem:structurepolytopes}
      Let $P=\Conv(\mathrm{A})\subset \mathbb R^{n}$ a full dimensional polytope with $N$ vertices and $\nuu=(\nu_0,\nu)\in\mathbb C^{n+1}$, with $\Re \nu_0 > 0$ as before. Consider further an index set of vertices $\sigma \subset \{1,\ldots,N\}$ which corresponds to a simplex, $\det \Aas \neq 0$. Then the following equivalence holds
      \begin{align}
        \nu/\nu_0\in P \Leftrightarrow (\Aas^{-1})_i \nuu \geq \sum_{j=1}^{N-n-1} (\Aas^{-1}\Aabs)_{ij} r_j \qquad \textrm{with} \qquad r\in\mathbb{R}_{\geq 0}^{N-n-1} 
      \end{align}
      where $ \nu/\nu_0:= (\nu_1/\nu_0,\ldots \nu_n/\nu_0)$ is the componentwise division. 
    \end{lemma}
    \begin{proof}
      From the vertex definition of polytopes it is
      \begin{align}
        \nu/\nu_0 \in P \Leftrightarrow \begin{pmatrix} 1 \\ \nu/\nu_0 \end{pmatrix} = \Aa \tilde k \Leftrightarrow \nuu = \Aa k
      \end{align}
      for $k,\tilde k\in \mathbb{R}_{\geq 0}^{N}$. Since the polytope is full dimensional, there is a triangulation and one can divide the polytope in $\Aa = (\Aas,\Aabs)$ where $\Aas$ has an inverse. Thus the lemma follows by multiplying by the inverse $\Aas^{-1}$.
    \end{proof}
    
    This lemma establishes a connection between the signs of the matrix $\Aas^{-1} \Aabs$ and the geometry of the polytope. Since the signs of $\Aas^{-1} \Aabs$ determine the kinematic regions where the series representations (theorem \ref{thm:FeynSeries}) converge fast, this lemma can be useful to analyze triangulations in a numerical application.
 
    Furthermore, the matrices $\Aas$ and $\Aabs$ which appear in the triangulation procedure has some special properties, which also guarantees the convergence of $\Gamma$-series.
    
    \begin{lemma} \label{lem:AA1} 
      Let $P=\Conv(\mathrm A)$ be a full dimensional polytope with integer vertices $\mathrm A \in\mathbb Z^{n\times N}$ and $T$ a triangulation, where the indices of the simplices are denoted by $\sigma$. Then it holds
      \begin{align}
        \sum_{i\in\sigma} (\Aas^{-1}\Aabs)_{ij} = 1 \label{eq:AA1a}
      \end{align}
      for all $j=1,\ldots,N$ and
      \begin{align}
        \sum_{i\in\sigma} (\Aas^{-1}\nuu)_i = \nu_0 \label{eq:AA1b}
      \end{align}
      where $\bar\sigma$ is the complement of $\sigma$ and $\nuu = (\nu_0,\nu_1,\ldots,\nu_n)\in\mathbb C^{n+1}$ an arbitrary complex vector.
    \end{lemma}
    \begin{proof}
      The columns of the matrix $B=\begin{pmatrix} -\Aas^{-1}\Aabs \\ \mathbbm{1}\end{pmatrix}\in\mathbb{Q}^{N\times r}$ are a basis of the lattice $\mathbb L := \ker (\Aa) \cap \mathbb Z^N$, which can be verified by direct computation $\Aa B = 0$. It follows particularly that $\sum_j B_{jk}=0$ and thereby equation (\ref{eq:AA1a}). The second statement (\ref{eq:AA1b}) follows trivially from $\sum_{j,k} (\Aas)_{ij} (\Aas^{-1})_{jk} \nuu_k=\nuu_i$.
    \end{proof}
    
    This lemma also implies that the appearing variables $x_j = (z_{\bar\sigma})_j \prod_i (z_\sigma)_i^{-(\Aas^{-1} \Aabs)_{ij}}$ in the $\Gamma$-series are dimensionless (without units). Therefore, the unit of the $\Gamma$-series will be only determined by the monomial in front of the power series and is equal to $[\phi_{\sigma,k}(\nuu,z)] = [z_i]^{-\nu_0}$. \spacepar

    To finish this short section about polytopes, we introduce another method to triangulate convex polytopes. In section \ref{ssec:interlude} there was already a method described how to construct triangulations by generic height vectors. Here, we discuss a combinatorial way to triangulate polytopes. This combinatorial approach also allows the structure of subtriangulations in a natural way.
    
    In order to ensure the convexity of polytopes, we will formulate the concept of visibility in mathematical terms. Consider a face $F$ of a convex polytope $P\subset \mathbb R^n$ and an arbitrary point $x\in\relint (F)$ in the relative interior of $F$. The face $F$ is \textit{visible} from another point $p\notin F$, if the line segment $[x,p]$ intersects $P$ only at $x$ \cite{DeLoera.Rambau.SantosTriangulationsStructuresAlgorithms2010}.
    
    With the concept of visibility, we can add points to a triangulation and obtains an enlarged triangulation. Let $T$ be a regular triangulation of the convex polytope $\Conv(\mathrm A)$. Then the set
    \begin{align}
      T^\prime = T \cup \{ B \cup \{p\} \, | \,  B \in T \textrm{ and } B \textrm{ is visible from } p \} 
    \end{align}
    is a regular triangulation of the convex polytope $\Conv(\mathrm A \cup p)$ \cite{DeLoera.Rambau.SantosTriangulationsStructuresAlgorithms2010}. In doing so, one can create a so-called \textit{placing triangulation} by starting with an arbitrary point (the triangulation of one point is still the point) and placing then step by step the other points to the previous triangulation. The order of the added points will determine the triangulation. \spacepar
    
    In practice there are different algorithms available to calculate triangulations of convex polytopes efficiently, e.g. TOPCOM \cite{Rambau.topcom} and polymake \cite{polymake:2000}.

    \subsection{Convergence of $\Gamma$-series} \label{ssec:app-convergence}
    
    The $\Gamma$-series were introduced in section \ref{ssec:GKZHyp} as formal solutions of the GKZ hypergeometric system. Using standard arguments, we proof, that those $\Gamma$-series always have a non-vanishing region of convergence. The proof is roughly orientated towards \cite{StienstraGKZHypergeometricStructures2005}. In order to show the convergence of the $\Gamma$-series one has to estimate the summands. As an application of the Stirling formula one can state the following lemma:
    
    \begin{lemma}[similar to \cite{StienstraGKZHypergeometricStructures2005}] \label{lem:gammaest}
      For every $C\in\mathbb C$ there are constants $\kappa,R\in\mathbb R_{>0}$ independent of $M$, such that
      \begin{align}
        \frac{1}{|\Gamma (C+M)|} \leq \kappa R^{|M|} |M|^{-M} \label{eq:gammest}
      \end{align}
      for all $M\in\mathbb Z$.
    \end{lemma}
    \begin{proof}
      Firstly, consider the non-integer case $C\notin \mathbb{Z}$. For $M>0$ it is
      \begin{align}
        |\Gamma(C+M)| &= |\Gamma(C)| \prod_{j=0}^{M-1} |(C+j)| \geq |\Gamma(C)| \prod_{j=0}^{M-1} ||C|- j| \nonumber\\
        &= |\Gamma(C)| \prod_{j=1}^M \left|\frac{|C|-j+1}{j}\right| j \geq M! Q^M |\Gamma(C)|
      \end{align}
      where $Q=\min \left|\frac{|C|-j+1}{j}\right| > 0$ and we used a variation of the triangle inequality $|a+b| \geq  ||a| - |b||$. With Stirling's approximation one obtains further
      \begin{align}
        |\Gamma(C+M)| &\geq |\Gamma(C)| \sqrt{2\pi} Q^M M^{M+\frac{1}{2}} e^{-M} \geq |\Gamma(C)| \sqrt{2\pi} \left(\frac{Q}{e}\right)^M M^M \point
      \end{align}
      In contrast, for $M<0$ and the triangle inequality $|a-b| \leq |a| + |b|$ one obtains
      \begin{align}
        |\Gamma(C+M)| &= |\Gamma(C)| \prod_{j=1}^{|M|} \left|\frac{1}{C-j}\right| \geq |\Gamma(C)| \prod_{j=1}^{|M|} \left|\frac{1}{|C|+j}\right| \geq |\Gamma(C)| \prod_{j=1}^{|M|} \left|\frac{1}{|C|+|M|}\right| \nonumber\\
        &= \left(1+\frac{|C|}{|M|}\right)^{-|M|} |M|^{M} |\Gamma(C)| \geq |\Gamma(C)| (1+|C|)^{-|M|} |M|^{M} \point
      \end{align}
      By the setting $\kappa = |\Gamma(C)|^{-1}$ and $R=\max \left( 1+|C|,e Q^{-1}\right)$ one can combine both cases to equation (\ref{eq:gammest}). The case $M=0$ is trivially satisfied whereat we set $0^0:=1$.
      
      Consider now the case where $C\in\mathbb Z$. If $C+M\leq 0$ the $\Gamma$-function has a pole and the lemma is trivial satisfied. For $C+M\geq 1$ it is $\Gamma(C+M) \geq \left(C+M-\frac{3}{2}\right) \left(C+M-\frac{5}{2}\right) \cdots \frac{1}{2} = \frac{\Gamma(C+M-\frac{1}{2})}{\Gamma(\frac{1}{2})} = \frac{1}{\sqrt \pi} \Gamma(C+M-\frac{1}{2})$ which recurs to the non-integer case with $C^\prime=C-\frac{1}{2}\notin \mathbb Z$.
    \end{proof}
    
    In order to estimate products of ``self-exponential'' functions $a^a$ the following inequality is helpful.  
  
    \begin{lemma} \label{lem:apowera}
      The following estimation holds for $a_i\in\mathbb{R}_{>0}$
      \begin{align}
        \left(\sum_{i=1}^N a_i\right)^{\sum_{i=1}^N a_i} \geq \prod_{i=1}^N a_i^{a_i} \geq \left(\frac{1}{N} \sum_{i=1}^N a_i\right)^{\sum_{i=1}^N a_i}
      \end{align}
      where $N\in\mathbb N$.
    \end{lemma}
    \begin{proof}
      The left inequality is trivially true. The right inequality will be proven by considering different cases of $N$. For $N=2$ without loss of generality it is $\frac{a_1}{a_2}:=r \geq 1$. The latter is then equivalent to the Bernoulli inequality
      \begin{align}
        a_1^{a_1} a_2^{a_2} \geq \left(\frac{a_1+a_2}{2}\right)^{a_1+a_2} \Leftrightarrow  \left(1+\frac{r-1}{1+r}\right)^{r+1} \geq r \point
      \end{align}
      For $N=2^n$ with $n\in\mathbb N$ the lemma can be reduced to the case $N=2$ by an iterative use. All the other cases can be reduced to a $2^n$-case by adding the mean value $\Delta := \frac{1}{N} \sum_{i=1}^N a_i$
      \begin{align}
        \left(\Delta^\Delta\right)^{2^n-N}  \prod_{i=1}^N a_i^{a_i} \geq \left(\frac{\sum_{i=1}^N a_i + (2^n-N)\Delta}{2^n}\right)^{\sum_{i=1}^N a_i + (2^n-N)\Delta} = \left(\frac{\sum_{i=1}^N a_i}{N}\right)^{\sum_{i=1}^N a_i} \Delta^{(2^n-N)\Delta}
      \end{align}
      with $2^n-N>0$.
    \end{proof}
  
    Now, one can estimate the summands of the $\Gamma$-series to find a convergent majorant.
  
    \begin{theorem}[Convergence of the $\Gamma$-series \cite{Gelfand.Graev.ZelevinskiHolonomicSystemsEquations1987}] \label{thm:GammaConverge}
      There is always a non-vanishing interval $I=(a,b)\neq \emptyset$ with $a,b\in\mathbb R^N_{>0}$ such that the $\Gamma$-series of equation (\ref{eq:GammaPowerSeriesVarphi}) converges absolutely for $|z|\in I$.
    \end{theorem}
    \begin{proof}
      The sum in equation (\ref{eq:GammaPowerSeriesVarphi}) can be written in the form
      \begin{align}
        \sum_{\lambda\in\Lambda_k} \frac{x^\lambda}{\Gamma(a+C\lambda)}
      \end{align}
      where  $x_j = \frac{(z_{\bar\sigma})_j}{\prod_i (z_\sigma)_i^{(\Aas^{-1} \Aabs)_{ij}}}$, $a := (1,\ldots,1,1-\Aas^{-1}\beta)\in\mathbb C^{N}$, $C:=\begin{pmatrix} \mathbbm{1}_r \\ -\Aas^{-1}\Aabs\end{pmatrix}\in\mathbb Q^{N\times r}$ and $r=N-n-1$. Due to the form of $\Lambda_k$ it is $C \lambda\in\mathbb Z^{N\times r}$.
      
      Furthermore by lemma \ref{lem:AA1} it is also $\sum_{i=1}^N C_{ij} = 0$, which will imply $D:=\sum_{C_{ij}>0} C_{ij}\lambda_j = - \sum_{C_{ij}>0} C_{ij}\lambda_j = \frac{1}{2} |\sum_{ij} C_{ij}\lambda_j| $. Thus, one can estimate by lemma \ref{lem:gammaest}
      \begin{align}
        \left| \prod_{i=1}^N \frac{1}{\Gamma(a_i+\sum_{j=1}^r C_{ij} \lambda_j)}\right| \leq K \prod_{i=1}^N R_i^{|\sum_{j=1}^r C_{ij}\lambda_j|} \left|\sum_{j=1}^r C_{ij}\lambda_j\right|^{-\sum_{j=1}^r C_{ij} \lambda_j} \point
      \end{align}
      With lemma \ref{lem:apowera} one can estimate further
      \begin{align}
       \prod_{i=1}^N \left|\sum_{j=1}^r C_{ij}\lambda_j\right|^{-\sum_{j=1}^r C_{ij} \lambda_j} \leq D^D N_2^D D^{-D} = N_2^D
      \end{align}
      by splitting the product in the $N_1$ factors with positive exponent $-\sum_{j=1}^r C_{ij}\lambda_j>0$ and the $N_2$ factors with negative exponent $-\sum_{j=1}^r C_{ij}\lambda_j<0$. With $R_{\textrm{max}}=\max_i R_i$ we get our final estimation
      \begin{align}
        \left| \prod_{i=1}^N \frac{1}{\Gamma(a_i+\sum_{j=1}^r C_{ij} \lambda_j)}\right| \leq K N_2^D R_{\textrm{max}}^{2D} \point
      \end{align}
      Thus, the $\Gamma$-series can be estimated by a geometric series and there is always a non-vanishing region of absolute convergence.
    \end{proof}
    
    Therefore, $\Gamma$-series have always a non-vanishing region of convergence and provide actual solutions of the GKZ hypergeometric system. In addition to it, it was shown in \cite{Gelfand.Graev.RetakhGeneralHypergeometricSystems1992} that all $\Gamma$-series of a triangulation have a common region of absolute convergence.

    \subsection{Characteristics of Several Simple Graphs}
    
    In order to classify different Feynman graphs with regard to their complexity from a hypergeometric perspective, we collate some characterizing numbers for several standard Feynman graphs in table \ref{tab:charac}. According to theorem \ref{thm:FeynSeries} one expect $\vol_0 \Delta_G$ series with dimension $r:=N-n-1$ in a hypergeometric series representation. Every of those series contains $(n+1)$ $\Gamma$-functions. Further, every regular, unimodular triangulation $\mathcal T_u$ gives a way to write the Feynman integral as Horn hypergeometric functions.
    
    If the number of master integrals $\mathfrak C = (-1)^n \chi((\mathbb C^\star)^n \setminus \{G=0\})$ \cite{Bitoun.Bogner.Klausen.PanzerFeynmanIntegralRelations2019} is less than the volume of the Newton polytope $\mathfrak C < \vol_0 \Delta_G$ one expects linear dependencies in the limit from the generalized Feynman integral to the ``ordinary'', non-generic Feynman integral according to section \ref{ssec:nongeneric}.

      \begin{sidewaystable}
	\begin{tabular}{l>{\ttfamily\footnotesize}c >{$}c<{$}>{$}c<{$}>{$}c<{$}>{$}c<{$}|>{$}c<{$}>{$}c<{$}|>{$}c<{$}>{$}c<{$}>{$}c<{$}|>{$}c<{$}>{$}c<{$}|>{$}c<{$}}
		  topology	&  \textrm{\normalsize Nickel index} & L	& n	& \text{\#legs} & \text{\#masses}	& N	& r	& \vol_0 \Delta_G & |\mathcal{T}|	& |\mathcal{T}_u| & |U| & |F| & \mathfrak C \\
		  \hhline{======|==|===|==|=}
		  bubble & e11|e:n00|n & 1 & 2 & 2 & 0 & 3 & 0 & 1 & 1 & 1 & 2 & 1 & 1\\
		  bubble & e11|e:nn0|n & 1 & 2 & 2 & 1 & 4 & 1 & 2 & 2 & 2 & 2 & 2 & 2\\
		  bubble & e11|e:nnn|n & 1 & 2 & 2 & 2 & 5 & 2 & 3 & 5 & 3 & 2 & 3 & 3\\
		  \hline
		  sunset & e111|e:n000|n & 2 & 3 & 2 & 0 & 4 & 0 & 1 & 1 & 1 & 3 & 1 & 1\\
		  sunset & e111|e:nn00|n & 2 & 3 & 2 & 1 & 6 & 2 & 3 & 6 & 6 & 3 & 3 & 2\\
		  sunset & e111|e:nnn0|n & 2 & 3 & 2 & 2 & 8 & 4 & 6 & 68 & 44 & 3 & 5 & 4 \\
		  sunset & e111|e:nnnn|n & 2 & 3 & 2 & 3 & 10 & 6 & 10 & 826 & 466 & 3 & 7 & 7\\
		  \hline
		  amputated cap & e112|2|e:n000|0|n & 2 & 4 & 2 & 0 & 8 & 3 & 5 & 20 & 20 & 5 & 3 & 2 \\
		  amputated cap & e112|2|e:n00n|0|n & 2 & 4 & 2 & 1 & 10 & 5 & 8 & 448 & 432 & 5 & 5 & 3\\
		  amputated cap & e112|2|e:n00n|n|n & 2 & 4 & 2 & 2 & 13 & 8 & 14 & 91052 & 43864 & 5 & 8 & 5 \\
		  \hline
		  vertex & e12|e2|e:n00|n0|n & 1 & 3 & 3 & 0 & 6 & 2 & 4 & 3 & 3 & 3 & 3 & 4\\
		  vertex & e12|e2|e:nn0|n0|n & 1 & 3 & 3 & 1 & 7 & 3 & 5 & 16 & 15 & 3 & 4 & 5\\
		  vertex & e12|e2|e:nnn|n0|n & 1 & 3 & 3 & 2 & 8 & 4 & 6 & 68 & 44 & 3 & 5 & 6\\
		  vertex & e12|e2|e:nnn|nn|n & 1 & 3 & 3 & 3 & 9 & 5 & 7 & 261 & 99 & 3 & 6 & 7\\
		  \hline
		  dunce's cap & e12|e22|e:n00|n00|n & 2 & 4 & 3 & 0 & 9 & 4 & 8 & 42 & 42 & 5 & 4 & 4\\
		  dunce's cap & e12|e22|e:nn0|n00|n & 2 & 4 & 3 & 1 & 11 & 6 & 11 & 2388 & 1968 & 5 & 6\\
		  dunce's cap & e12|e22|e:nnn|n00|n & 2 & 4 & 3 & 2 & 13 & 8 & 14 & 91052 & 43864 & 5 & 8\\
		  \hline		  
		  box & e13|e2|e3|e:n00|n0|n0|n & 1 & 4 & 4 & 0 & 10 & 5 & 11 & 102 & 102 & 4 & 6 & 11\\
		  box & e13|e2|e3|e:nn0|n0|n0|n & 1 & 4 & 4 & 1 & 11 & 6 & 12 & 1689 & 1260 & 4 & 7 & 12\\
		  box & e13|e2|e3|e:nnn|n0|n0|n & 1 & 4 & 4 & 2 & 12 & 7 & 13 & 14003 & 8004 & 4 & 8 & 13\\
		  box & e13|e2|e3|e:nnn|nn|n0|n & 1 & 4 & 4 & 3 & 13 & 8 & 14 & 87657 & 34143 & 4 & 9 & 14\\
		  box & e13|e2|e3|e:nnn|nn|nn|n & 1 & 4 & 4 & 4 & 14 & 9 & 15 & 469722 (192) & 114276 (12) & 4 & 10 & 15 \\
		  \hline
		  kite & e13|23|e3:n00|00|n0 & 2 & 5 & 2 & 0 & 16 & 10 & 42 & > 9 \cdot 10^6 &  > 9 \cdot 10^6 & 8 & 8 & 3\\
		  \hline
		  banana & e1111|e:n0000|n & 3 & 4 & 2 & 0 & 5 & 0 & 1 & 1 & 1 & 4 & 1 & 1\\
		  banana & e1111|e:nn000|n & 3 & 4 & 2 & 1 & 8 & 3 & 4 & 24 & 24 & 4 & 4\\
		  banana & e1111|e:nnn00|n & 3 & 4 & 2 & 2 & 11 & 6 & 10 & 2486 & 1618 & 4 & 7\\
		  banana & e1111|e:nnnn0|n & 3 & 4 & 2 & 3 & 14 & 9 & 20 & 522206 (3952) & 248420 (2194) & 4 & 10\\
		  banana & e1111|e:nnnnn|n & 3 & 4 & 2 & 4 & 17 & 12 & 35 & > 5 \cdot 10^6 & > 5 \cdot 10^6 & 4 & 13 & 15\\
  	\end{tabular}
  	\caption{Classification of some basic Feynman integrals, with loop number $L$, number of edges $n$, $N$ the number of monomials in $G$, $r$ the corank of $\Aa$, $\vol_0 \Delta_G$ the dimension of solution space, $|\mathcal T|$ the number of regular triangulations, $|\mathcal T_u|$ the number of regular unimodular triangulations, $|U|$ the number of monomials in the first Symanzik polynomial, $|F|$ the number of monomials in the second Symanzik polynomial and $\mathfrak C$ the number of master integrals. Numbers in brackets count non regular triangulations.} \label{tab:charac}
      \end{sidewaystable}

 \nocite{*}
  \pagebreak
 \printbibliography

\end{document}

